\documentclass[12pt]{article}
\usepackage{amssymb}
\usepackage{amsmath}
\usepackage{epsfig}
\usepackage{color}
\usepackage{amsmath,amssymb,amsthm,a4,bm,epsfig,color}

\newcommand{\ie}{i.\,e.}

\newtheorem{theorem}{Theorem}
\newtheorem{proposition}{Proposition}
\newtheorem{lemma}{Lemma}
\newtheorem{remark}{Remark}
\newtheorem*{remark*}{Remark}
\newtheorem*{remarks*}{Remarks}

\newcommand{\bbR}{{\mathbb R}}
\newcommand{\bbC}{{\mathbb C}}
\newcommand{\bbN}{{\mathbb N}}
\newcommand{\Hone}{H^{1}({\mathbb R}^N,{\mathbb C})}
\newcommand{\HmOne}{H^{-1}({\mathbb R}^N, {\mathbb C})}
\newcommand{\cL}{{\mathcal L}}
\newcommand{\cM}{{\mathcal M}}
\newcommand{\cN}{{\mathcal N}}
\newcommand{\cE}{{\mathcal E}}
\newcommand{\cT}{{\mathcal T}}
\newcommand{\cR}{{\mathcal R}}

\newcommand{\cC}{{\mathcal C}}

\newcommand{\sY}{{\mathsf Y}}
\newcommand{\cF}{{\mathcal F}}
\newcommand{\esssup}{ess\,sup}
\renewcommand{\Re}{{\mathrm{Re}}}
\renewcommand{\Im}{{\mathrm{Im}}}

\begin{document}

\title{Solitary wave dynamics in time-dependent potentials}

\author{Walid K. Abou Salem\footnote{Department of Mathematics, University of Toronto, Toronto, Ontario, Canada M5S 2E4} \footnote{E-mail: walid@math.utoronto.ca}}
\date{}

\maketitle

\begin{abstract}

We rigorously study the long time dynamics of solitary wave solutions of the nonlinear Schr\"odinger equation in  {\it time-dependent} external potentials. To set the stage, we first establish the well-posedness of the Cauchy problem for a generalized nonautonomous nonlinear Schr\"odinger equation. We then show that in the {\it space-adiabatic} regime where the external potential varies slowly in space compared to the size of the soliton, the dynamics of the center of the soliton is described by Hamilton's equations, plus terms due to radiation damping. We finally remark on two physical applications of our analysis. The first is adiabatic transportation of solitons, and the second is Mathieu instability of trapped solitons due to time-periodic perturbations.

\end{abstract}


\section{Introduction}

\subsection{Heuristic discussion and overview of earlier results}

In the last few years, there has been substantial progress in rigorously understanding solitary wave dynamics of the nonlinear Schr\"odinger equation in {\it time-independent} external potentials. The basic picture is that in the semi-classical limit, the dynamics of the center of the soliton is described by Hamilton's (or Newton's) equations, plus terms due to {\it radiation damping}. This is a beautiful example where the solitary wave solution of the nonlinear Schr\"odinger equation behaves like a classical point particle in a suitable limit. Rigorous confirmation of this picture in time-independent potentials has been given  in \cite{FTY1,FTY2} for the Hartree equation, and in \cite{BJ1} for local nonlinearities. The case of general nonlinearities has been studied in \cite{FJGS1} ; see also \cite{FJGS2,GS1,HZ1}. 

In this paper, we rigorously study the dynamics of solitary wave solutions of the nonlinear Schr\"odinger equation in {\it time-dependent} external potentials. As far as we know, this problem has not been studied in the literature so far. We show that when the potential varies {\it slowly} in space compared to the size of the soliton, which we call the {\it space-adiabatic regime}, the center of mass motion of the soliton is almost like that of a classical point particle in the external potential,  independent of the rate of change of the potential with time. \footnote{Unlike the case when the external potential is time-independent, the semi-classical limit and the space-adiabatic limit are not equivalent when the potential is time-dependent. Instead, the semi-classical limit, which is equivalent to scaling both time and space $t\rightarrow ht, x\rightarrow hx,$ $h\ll 1,$ is a special case of the space-adiabatic limit (where only space scales as $x\rightarrow hx, h\ll 1$).  These statements will be made precise in the following sections.} We also show that this picture holds for much longer time scales ($O(|\log h|/h), \ \ h\ll 1,$) than the one given in \cite{FJGS1} ($O(h^{-1}), \ \ h\ll 1$). Along the way, we discuss sufficient conditions for the well-posedness of a generalized nonautonomous nonlinear Schr\"odinger equation with time-dependent potentials and nonlinearities. We finally sketch two physical applications of our analysis. The first is adiabatic transportation of solitons, and the second Mathieu instability of trapped solitons due to time-periodic perturbations. Our analysis relies on important developments in the nonlinear Schr\"odinger equation during the past two decades, \cite{Str1,St1,GV1,GV2,BL1,BL2,BLP1,GSS1,GSS2,We1,We2,BP1,BS1,Ka1,Ka2,BJ1,FTY2}, particularly \cite{FJGS1}; see also for \cite{Ca1,SS1} for comprehensive reviews.

We note that our analysis regarding the well-posedness of the generalized nonautonomous nonlinear Schr\"odinger equation is of some independent interest, especially that rigorous investigation of such equations is poor compared to the autonomous one, despite the former's relevance to very many experiments in quantum optics and Bose-Einstein condensates, where experimentalists can change various parameters with time.

\subsection{Notation}


\noindent In the following, $L^p(I)$  denotes the standard Lebesgue space, $1\le p\le \infty,$ with norm
\begin{equation*}
\|f\|_{L^p} = (\int_I dx~ |f(x)|^p)^{\frac{1}{p}}, \ \ f\in L^p(I),  p<\infty , \ \ \|f\|_{L^\infty} = \esssup(|f|) , \ \ f\in L^\infty (I).
\end{equation*}
We also define
\begin{equation*}
\|f\|_{L^p(I,L^q(J))}:= \|\, \|f\|_{L^q(J)}\, \|_{L^p(I)}.
\end{equation*}
For $1\le p\le \infty,$ $p'$ is the conjugate of $p, \ie , 1/p+1/p' = 1.$
We denote by $\langle \cdot ,\cdot \rangle $ the scalar product in $L^2(\bbR^N),$ 
\begin{equation*}
\langle u, v\rangle = \Re\int_{\bbR^N} u\overline{v}, \ \ u,v\in L^2(\bbR^N),
\end{equation*}
as well as its extension by duality to ${\mathsf Y}\times {\mathsf Y}',$ where ${\mathsf Y}$ and ${\mathsf Y}'$ are complete metric spaces such that ${\mathsf Y}\hookrightarrow L^2 \hookrightarrow {\mathsf Y}',$ with dense embedding

\noindent Given the multi-index $\overline{\alpha}= (\alpha_1,\cdots ,\alpha_N) \in \bbN^N,$ we denote $|\overline{\alpha}| = \sum_{i=1}^N \alpha_i.$ Furthermore, $\partial_x^{\overline{\alpha}} := \partial_{x_1}^{\alpha_1}\cdots \partial_{x_N}^{\alpha_N}.$

\noindent 
For $1\le p\le\infty$ and $s\in\bbN,$ the (complex) Sobolev space is given by
\begin{equation*}
W^{s,p}(\bbR^N) := \{u\in {\mathcal S}'(\bbR^N): \partial_x^{\overline{\alpha}} u\in L^p(\bbR^N), |\overline{\alpha}|\le s\}, 
\end{equation*}
where ${\mathcal S}'(\bbR^N)$ is the space of tempered distributions. We equip $W^{s,p}$ with the norm 
\begin{equation*}
\|u\|_{W^{s,p}} = \sum_{\alpha, |\mathbf{\alpha}|\le s} \|\partial_x^{\overline{\alpha}} u \|_{L^p},
\end{equation*}
which makes it a Banach space. Moreover, $W^{-s,p'}$ is the dual of $W^{s,p}.$

\noindent We denote by 
\begin{equation*}
H^{s,p}(\bbR^N):= \{u\in {\mathcal S}'(\bbR^N): \cF^{-1} (1+|k|^2)^{\frac{s}{2}}\cF u \in L^p(\bbR^N)\}, \ \ s\in\bbR, 1\le p\le \infty,
\end{equation*}
where ${\mathcal F}$ stands for the Fourier transform. The space $H^{s,p}$ is equipped with the norm
\begin{equation*}
\|u\|_{H^{s,p}} = \|{\mathcal F}^{-1}(1+|k|^2)^{\frac{s}{2}} {\mathcal F} u\|_{L^p} , \ \ u\in H^{s,p}(\bbR^N),
\end{equation*}
which makes it a Banach space.  We use the shorthand $H^{s,2}=H^s.$

\noindent Given $f$ and $g$ real functions on $\bbR^N,$ we denote their convolution by $\star,$
\begin{equation*}
f\star g(x):= \int dy~ f(y-x) g(y) .
\end{equation*}

\noindent Given $x\in \bbR^N,$ we denote $\|x\|:=\sqrt{\sum_{i=1}^N x_i^2}.$

\subsection{Description of the problem}

In this paper, we study the long time dynamics of solitary wave solutions of the nonlinear Schr\"odinger equation in time-dependent external potentials. The nonlinear Schr\"odinger equation is of the form 
\begin{equation}
\label{eq:NLSE}
i\partial_t \psi = (-\Delta + V_h(t,x))\psi -f(\psi),
\end{equation}
where $\psi: {\mathbb R}\times{\mathbb R}^N \rightarrow {\mathbb C}, x\in {\mathbb R}^N$ denotes a point in the configuration space, $t\in {\mathbb R}$ is time, $\partial_t = \frac{\partial}{\partial t}, \Delta = \sum_{j=1}^N \frac{\partial^2}{\partial_{x_j}^2}$ the $N$-dimensional Laplacian, $V_h$ is the external potential, such that 
\begin{equation*}
V_h(t,x)\equiv V(t,hx) , h\in {\mathbb R}^+,
\end{equation*}
and the nonlinearity $f$ is a mapping on complex Sobolev spaces such that 
\begin{equation*}
f:H^1({\mathbb R}^N, {\mathbb C}) \rightarrow H^{-1}({\mathbb R}^N, {\mathbb C}), 
\end{equation*}
$f(0)=0,$ and $\overline{f(\psi)}=f(\overline{\psi}),$ where $\overline{\cdot}$ denotes complex conjugation. Typical nonlinearities are local ones
\begin{equation}
\label{eq:LocalNL}
f(\psi)=\lambda |\psi|^s \psi, \ \ \lambda>0,\ \  0<s<\frac{4}{N},
\end{equation} 
and Hartree (nonlocal) nonlinearites
\begin{equation*}
f(\psi)=\lambda (W\star |\psi|^2)\psi ,\lambda>0 ,
\end{equation*}
where $W$ is continuous, positive, spherically symmetric, and tends to zero as $\|x\|\rightarrow \infty.$ Note that the above nonlinearities are self-focusing, and a general nonlinearity can be a sum of both local and nonlocal ones. The external potentials that we consider in this paper satisfy 
\begin{equation}
\label{eq:Pot1}
V_h(t,x) = V(t,hx), \ \ V(t,x) \in C^1(\bbR, C^2(\bbR^N))
\end{equation}
such that 
\begin{equation}
\label{eq:Pot2}
\partial_x^{\overline{\alpha}} V \in L^\infty(\bbR,L^\infty(\bbR^N)+ L^p(\bbR^N)) , \ \ |\overline{\alpha}|\le 1,
\end{equation}
$p> \frac{N}{2}, p\ge 1,$
and 
\begin{equation}
\label{eq:Pot3}
\partial_t V_h(t,x)\in L^\infty(\bbR,L^\infty(\bbR^N)) .
\end{equation}
General assumptions about the model will be discussed in detail in Subsection \ref{sec:Assumptions}. We also show that the Cauchy problem with these assumptions is (globally) well-posed in $H^1$ in Section \ref{sec:Well-Posedness}. 

When $V=0,$ the nonlinear Schr\"odinger equation (\ref{eq:NLSE}) with nonlinearities as given above admits solitary wave solutions, which are stable stationary, spherically symmetric and positive solutions
\begin{equation}
 \eta_{\sigma}(x,t):= 
    e^{i(\frac{1}{2}v\cdot(x-a)+\gamma)}\eta_\mu(x-a),
   \label{eq:Sol}
\end{equation}
where $\sigma:=\{a,v,\gamma,\mu\}$, 
$a=vt+a_0$, $\gamma=\mu t +  
\frac{v^2}{4}t+ \gamma_{0}$, with $\gamma_0 \in 
[0,2\pi)$, $a_0,v\in\bbR^N$ and
$\mu\in\mathbb{R}^{+}$, constant, 
and $\eta_\mu$ is a positive solution of the 
nonlinear eigenvalue problem
\begin{equation}
(-\Delta + \mu )\eta_\mu -f(\eta_\mu)=0.
\end{equation}
The solution (\ref{eq:Sol}) stands for a solitary travelling wave with velocity $v,$ center $a$ and phase $\frac{1}{2}v\cdot(x-a)+\gamma,$ and the {\it size} of the soliton is $\propto \mu^{-1/2},$ in the sense that $\eta_\mu \sim e^{-\sqrt{\mu}\|x\|}$ as $\|x\|\rightarrow\infty,$ see \cite{St1,GV1,GV2,BL1,BLP1,GSS1,GSS2,We1,We2,BP1,BS1,Cu1,Cu2,Ca1,SS1}. We consider in this paper potentials which vary slowly in space compared to the size of the soliton, $i.e.,$ 
\begin{equation*}
\sup_{t\in \bbR} \frac{Supp |\nabla V(x,t)|}{\sqrt{\mu}}\ll 1,
\end{equation*}
which corresponds to the space-adiabatic limit if we set the size of the soliton to $O(1).$ 

We now state a rigorous result for the special class
of nonlinearities discussed above. A more general result, Theorem \ref{th:Main}, will be stated in Subsection \ref{sec:Main} after listing general assumptions.

\begin{theorem} 
\label{th:Introduction}
Suppose the nonlinearity $f$ is given by (\ref{eq:LocalNL}), 
and that the external potential $V_h$ satisfies (\ref{eq:Pot1})-(\ref{eq:Pot3}). Let $I_0$ be any closed, bounded interval in $\bbR^+$. For $h\ll 1,$ suppose the initial condition $\psi_0$ satisfies 
\begin{equation*}
\|e^{-i\frac{1}{2}v_0\cdot x}(\psi_0-\eta_{\sigma_{0}})\|_{H^1}<h,
\end{equation*}
for some $\sigma_0\in \bbR^N\times\bbR^N\times [0,2\pi)\times I_0$. 
Then, for small enough $h\ll 1,$ there exists an absolute positive constant $C$, independent of $h,$ but possibly dependent on $I_0,$ such that for times $0\leq t\leq  C  \frac{|log (h)|}{h},$ the 
solution to the nonlinear Schr\"odinger equation (\ref{eq:NLSE}) with initial condition $\psi_0$ is of the form 
\begin{equation*}
\psi(x,t) =e^{i(\frac{1}{2}v\cdot (x-a) + \gamma)}(\eta_\mu(x-a) + w(x-a,t)),
\end{equation*}
where 
\begin{equation*}
\|w\|_{H^1}=O(h^{\frac{3}{4}}),
\end{equation*}
and where the parameters $v, a, \gamma$ and $\mu$ satisfy the 
differential equations
\begin{align*}
\frac{1}{2}\partial_t v & = - (\nabla V)(t,a) + O(h^{\frac{3}{2}}), \\ 
\partial_t a & = v + O(h^{\frac{3}{2}}), \\
\partial_t \gamma & = \mu - V(t,a) + \frac{1}{4}v^2 + O(h^{\frac{3}{2}}), \\ 
\partial_t{\mu} &= O(h^{\frac{3}{2}}).
\end{align*}
\end{theorem}
In other words, for initial conditions close enough to a 
solitary wave solution, and for external potentials which vary slowly compared to the size of the soliton, the center of mass motion of the solitary wave is determined by Hamilton's (or Newton's) equations
of motion for a point particle in the  external potential, up to
small corrections corresponding to radiation damping. The same result holds for more general nonlinearities, see Section \ref{sec:MainSect}.

The organization of this paper is the following. In Section \ref{sec:MainSect}, we list the assumptions on the nonlinearity and the potential and state the main result of the paper. We also discuss models where the various assumptions are satisfied. In Section \ref{sec:Well-Posedness}, we discuss the well-posedness of a generalized nonautonomous nonlinear Schr\"odinger equation, such that (\ref{eq:NLSE}) corresponds to the special case when only the potential is time-dependent. We then recall basic useful properties of the nonlinear Schr\"odinger equation and the soliton manifold in Section \ref{sec:NLSeq}. In Section \ref{sec:ProofMain}, we prove the main result of the paper, Theorem \ref{th:Main}. Finally, in Section \ref{sec:Applications}, we discuss two physical applications of our analysis. The first is the adiabatic transportation of solitons, and the second is Mathieu instability of trapped solitons due to time-periodic perturbations. 

\vspace{0.5cm}
\noindent {\bf Acknowledgements}

I thank the Center for Theoretical Studies at ETH Zurich, particularly J\"urg Fr\"ohlich, for their very kind hospitality. The partial financial support of NSERC grant NA 7901 is gratefully acknowledged.

\section{Main Result}\label{sec:MainSect}

In this section, we precisely state the main result of this paper after listing our assumptions.

\subsection{The Model}\label{sec:Assumptions}

We now list our assumptions and discuss models where they are satisfied.

\begin{itemize}
\item[(A1)] {\it Nonlinearity.} 
The nonlinearity $f=f_1+\cdots+f_k$ such that 
\begin{equation*}
f_j\in C^2(\Hone,\HmOne), \ \ j=1,\cdots ,k,
\end{equation*} 
satisfy the following. 
\begin{equation*}
\exists F_j \in C^3(\Hone,\bbR),
\end{equation*}   
with $F'_j=f_j,$ where the prime stands for the Fr\'echet derivative. 
\begin{equation*}
\exists r_j\in [2,\frac{2N}{N-2}), \ \ ([2,\infty], N=1),
\end{equation*} 
such that $\forall M>0, \exists$ a finite constant $C_j(M)$ such that 
\begin{equation*}
\| f_j(u)-f_j(v)\|_{L^{r_j'}}\le C_j(M)\|u-v\|_{L^{r_j}},
\end{equation*}
$\forall u,v\in \Hone, \|u\|_{H^1}+\|v\|_{H^1}\le M.$ Furthermore, 
\begin{equation*}
\Im f_j(u)\overline{u}=0
\end{equation*}
almost everywhere on $\bbR^N, \forall u\in \Hone.$
Let $F:= \sum_{j=1}^k F_j.$ For every $M>0,$ there exists a positive constant $C(M)$ and $\epsilon \in (0,1)$ such that
\begin{equation*}
F(u)\le \frac{1-\epsilon}{2}\|u\|_{H^1} + C(M), \forall u\in H^1(\bbR^N)
\end{equation*}
such that $\|u\|_{L^2}\le M.$ Furthermore, 
\begin{align*}
&\sup_{\|u\|_{H^1}\le M} \|F''(u)\|_{{\mathcal B}(H^1,H^{-1})} < \infty\\
&\sup_{\|u\|_{H^1}\le M} \|F'''(u)\|_{H^1\rightarrow {\mathcal B}(H^1,H^{-1})} \le \infty,
\end{align*}
where ${\mathcal B}$ denotes the space of bounded operators.

\item[(A2)] {\it External Potential.}
The external potential satisfies
\begin{equation*}
V_h(t,x) = V(t,hx), \ \ V(t,x) \in C^1(\bbR, C^2(\bbR^N))
\end{equation*}
such that 
\begin{equation*}
\partial_x^{\overline{\alpha}} V \in L^\infty(\bbR,L^\infty(\bbR^N)+ L^p(\bbR^N)) , \ \ |\overline{\alpha}|\le 1,
\end{equation*}
$p> \frac{N}{2}, p\ge1,$
and 
\begin{equation*}
\partial_t V_h(t,x)\in L^\infty(\bbR,L^\infty(\bbR^N)) .
\end{equation*}

\item[(A3)]{\it Symmetries.} The nonlinearity $F$ satisfies $F(T\cdot)=F(\cdot),$ where $T$ is a translation 
\begin{equation*}
T_a^{tr}: u(x)\rightarrow u(x-a), \ \ a\in \bbR^N,
\end{equation*}
a rotation 
\begin{equation*}
T_R^r: u(x)\rightarrow u(R^{-1}x), \ \ R\in SO(N),
\end{equation*}
a gauge transformation
\begin{equation*}
T_\gamma^g : u(x)\rightarrow e^{i\gamma}u(x), \ \ \gamma \in [0,2\pi),
\end{equation*}
a boost 
\begin{equation*}
T_v^b: u(x)\rightarrow e^{\frac{i}{2}v\cdot x} u(x), \ \ v\in \bbR^N
\end{equation*}
or a complex conjugation
\begin{equation*}
T^c: u(x)\rightarrow \overline{u}(x).
\end{equation*}

\item[(A4)]{\it Solitary Wave.} $\exists I\subset \bbR$ such that $\forall \mu\in I,$ the nonlinear eigenvalue problem 
\begin{equation*}
(-\Delta +\mu)\eta_\mu -f(\eta_\mu)=0
\end{equation*}
has a positive, spherically symmetric solution $\eta_\mu\in L^2(\bbR^N)\cap C^2(\bbR^N),$ such that 
\begin{equation*}
\||x|^3\eta_\mu\|_{L^2}+\||x|^2|\nabla \eta_\mu|\|_{L^2} + \||x|^2\partial_\mu \eta_\mu\|_{L^2}<\infty, \forall \mu \in I.
\end{equation*}

\item[(A5)]{\it Orbital Stability.} The solution $\eta_\mu$ appearing in assumption (A4) satisfies
\begin{equation*}
\partial_\mu \int dx ~ \eta_\mu^2 >0, \forall \mu\in I.
\end{equation*} 

\item[(A6)] {\it Null Space Condition.} We define
\begin{equation*}
\cL_\mu := -\Delta +\mu -f'(\eta_\mu),
\end{equation*}
which is the Fr\'{e}chet derivative of the map $\psi\rightarrow (-\Delta+\mu)\psi -f(\psi)$ evaluated at $\eta_\mu.$ For all $\mu\in I,$ the null space
\begin{equation*}
\cN(\cL_\mu) = span\{i\eta_\mu, \partial_{x_j}\eta_\mu, j=1,\cdots, N\}.
\end{equation*}

\end{itemize}

\begin{remark}\label{rm:Nonlinearity}
Assumptions (A1) and (A2) are sufficient to establish global well-posedness of the nonlinear Schr\"odinger equation in $H^1,$ see Section \ref{sec:Well-Posedness}. Moreover, (A1) implies 
\begin{align*}
&|F(u+v)-F(u)-\langle F'(u), v\rangle | \le C(M) \|v\|_{H^1}^2 \\
&|F(u+v)-F(u)-\langle F'(u), v\rangle- \frac{1}{2}\langle F''(u)v,v\rangle| \le C(M)\|v\|_{H^1}^3 \\
&\|F'(u+v)-F'(u)-F''(u)v\|_{H^{-1}}\le C(M)\|v\|_{H^1}^2 ,
\end{align*}
for any $u\in H^2(\bbR^N)$ and $v\in H^1(\bbR^N)$ such that $\|u\|_{H^1}+\|v\|_{H^1}\le M.$ 
\end{remark}

\begin{remark}
 Assumption (A1) is satisfied, for example, if 
\begin{equation*}
F(u)= \frac{1}{2}\int dx ~ G(|u|^2) + W\star|u|^2,
\end{equation*}
where $G(r)=\int_0^r ds g(s),$ such that $g\in C^2(\bbR^+)$ with $g(s)\le C(1+s^\alpha),\alpha\in [0,\frac{2}{N}),$ $|\partial_s^k g(s)|\le C(1+s^{q-k}), k=0,1,2, q\in [0,\frac{2}{N-2}), N\ge 3,$ $q\in [0,\infty), N=1,2,$ and $W\in L^p+L^\infty, p>\frac{N}{2}, p\ge 1,$ such that $\max(0,W)\in L^r(\bbR^N)+L^\infty(\bbR^N), r> \frac{N}{2},\ge 1, N\ge 2,$ see \cite{Ca1}. 

Assumption (A3) follows if $W(r)=W(|r|).$ Assumption (A4) is satisfied for local nonlinearities if 
\begin{align*}
   -\infty &< \lim_{s\rightarrow 0} g(s) < \mu \\
   -\infty &\le\lim_{s\rightarrow\infty}s^{-\alpha}g(s) \le C ,  
\end{align*}   
where $0<\alpha < 2/(N-2)$, when $N>2$ and
$\alpha \in(0,\infty)$ if  $N=1,2,$ such that  
\begin{equation*}
\exists \zeta>0, \ \text{such that}\ 
        \int^\zeta_0ds  g(s) >\mu\zeta,
\end{equation*}
see for example \cite{BL1,BL2,BLP1,St1,Ca1}. Moreover, (A4) is satisfied for nonlocal nonlinearities if, in addition to the above, 
\begin{equation*}
    W\in L^{q}_{\mathrm{loc}},\ q\ge\frac{N}{2},\ 
W\rightarrow 0\ \text{as}\ \|x\|\rightarrow \infty;
\end{equation*}
see \cite{Ca1,GV2,FTY1,FTY2}. Assumption (A5) imply orbital stability of the solitary wave solution, see \cite{GSS1,GSS2}. It is satisfied for local nonlinearities $f(\psi)=\lambda |\psi|^{s}\psi, s<\frac{4}{N}.$ 
Assumption (A6) is satisfied for local nonlinearities if 
\begin{equation*}
g'(s)+g''(s)s^2 >0,
\end{equation*}
or if $N=1,$ \cite{FJGS1, We1}.

\end{remark}


\subsection{Statement of the Main Result}\label{sec:Main}

In this subsection, we state the main result of this paper, which will be proven in Section \ref{sec:ProofMain}.

\begin{theorem}\label{th:Main}
Suppose assumptions (A1)-(A7) hold.  Given $h>0$ such that $h\ll 1,$ suppose there exists $\sigma_0=\{a_0,v_0,\gamma_0,\mu_0\}\in \bbR^N\times\bbR^N\times [0,2\pi)\times I$ such that the initial condition $\psi_0\in H^1$ with 
\begin{equation*}
\| e^{-\frac{i}{2}v_0\cdot x}(\psi_0-\eta_{\sigma_0})\|_{H^1}<h.
\end{equation*}
Fix $\epsilon \in (0,1).$
Then, for $h\ll 1,$ there exists an absolute positive constant $C,$ which is independent of $h$ and $\epsilon,$ but which might depend on $\sigma_0,$ such that the solution to the nonlinear Schr\"odinger equation (\ref{eq:NLSE}) with initial condition $\psi_0$ can be written, for all time $t\in [0,C \epsilon |\log h|/h),$ as
\begin{equation*}
\psi(x,t)= e^{\frac{i}{2}v\cdot (x-a)+i\gamma} (\eta_\mu (x-a)+w(x-a,t)),
\end{equation*}
where $\|w\|_{H^1}=O(h^{1-\frac{\epsilon}{2}}),$ and the parameters $a,v,\gamma$ and $\mu$ satisfy the differential equations
\begin{align*}
\partial_t a &= v+ O(h^{2-\epsilon}) \\
\partial_t v &= -2 (\nabla V)(t,a) + O(h^{2-\epsilon})\\
\partial_t \gamma &= \mu -V(t,a)+\frac{1}{4}v^2 + O(h^{2-\epsilon})\\
\partial_t \mu &= O(h^{2-\epsilon}).
\end{align*}

\end{theorem}

We will prove this theorem in Section \ref{sec:ProofMain}. We now discuss the well-posedness of a generalized nonautonomous nonlinear Schr\"odinger equation.


\section{Well-posedness of a generalized nonautonomous nonlinear Schr\"odinger equation}\label{sec:Well-Posedness}

In this section, we discuss the local and global well-posedness of a generalized nonautonomous nonlinear Schr\"odinger equation with time-dependent nonlinearities and potential. We treat the potential and the nonlinearity as time-dependent perturbations. The application in this paper corresponds to the special case when only the external potential is time-dependent.

Consider the problem corresponding to a generalized nonlinear Schr\"odinger equation
\begin{equation}
\label{eq:NANLSE}
i\partial_t \psi = -\Delta\psi + g(t,\psi), \ \ \psi(t=0)=\phi,
\end{equation}
where $g$ contains both the potential and the nonlinearity. Note that $g$ can also depend on $x\in\bbR^N,$ but we drop the explicit dependence when there is no danger of confusion. In what follows, we say that $(q,r)$ is an admissible pair if 
\begin{align}
r &\in [2,\frac{2N}{N-2}), \ \ (r\in [ 2,\infty], N=1) \nonumber \\
\frac{2}{q}&=N(\frac{1}{2}-\frac{1}{r}) \label{eq:Admissible}
\end{align}
We make the following assumptions on $g.$

\begin{itemize}

\item[(B1)] The nonlinearity $g=g_1+\cdots +g_k$ such that $g_j\in C(\bbR, C(H^1,H^{-1})), j=1,\cdots, k.$  
\item[(B2)] There exist admissible pairs $(q_j,r_j),$ $j=1,\cdots k,$ such that, for every $T,M>0,$ there exist a constant $C(M)$ independent of $T,$ and $\beta$ independent of $T$ and $M,$ such that
\begin{equation*}
\|g_j(t,u)-g_j(t,v)\|_{L^{r_j'}(\bbR^N)} \le C(M)(1+T^{\beta})\|u-v\|_{L^{r_j}(\bbR^N)} ,
\end{equation*}
for all $u,v\in H^1$ with $\|u\|_{H^1}+\|v\|_{H^1}\le M,$ and $|t|<T,$ where $r'$ is the conjugate of $r, \ie , 1/r+1/r'=1.$ Furthermore, 
\begin{equation*}
\|g_j(t,u)\|_{W^{1,r_j'}} \le C(M) (1+T^\beta)  (1+\|u\|_{W^{1,r_j}})
\end{equation*}
for all $u\in H^1(\bbR^N)\cap W^{1,r}(\bbR^N)$ such that $\|u\|_{H^1}\le M$ and $|t|\le T.$

\item[(B3)] $\Im g_j(t,u)\overline{u}=0, j=1,\cdots , k,$ almost everywhere on $\bbR^N,$ for all $t\in \bbR$ and $u\in H^1.$

\item[(B4)] There exists a functional $G_j\in C(\bbR, C^1(H^1,\bbR))$ with $G_j'=g_j,$ where the prime stands for the Fr\'{e}chet derivative.  We let $G=G_1+\cdots G_k.$ For $u\in H^1,$  
\begin{equation}
\label{eq:GWP2}
|\partial_t G(t,u)|\le \tilde{C}(\|u\|_{L^2})l(t),
\end{equation}
where $\tilde{C}$ depends only on $\|u\|_{L^2}$ and the real function $l\in L^{\infty}(\bbR)$ such that $l(t)\le 1$ for almost all $t\in \bbR.$

\item[(B5)] For all $M>0,$ there exists $C(M)>0$ and $\epsilon\in (0,1),$ both independent of $t\in \bbR,$ such that 
\begin{equation}
\label{eq:GWP1}
|G(t,u)|\le \frac{1-\epsilon}{2}\|u\|_{H^1}^2 +C(M),
\end{equation}
uniformly in $t\in\bbR, \forall u\in H^1,$ such that $\|u\|_{L^2}\le M.$ 
\end{itemize}



We first prove local well-posedness of the Cauchy problem by extending Kato's method, which is based on Strichartz estimates and a fixed point argument, \cite{Ka1,Ka2}. Proving global well-posedness for data which are not necessarily small is a little bit more delicate, since energy is not conserved. 

\begin{proposition}\label{pr:LWP}
Suppose $g$ satisfies assumptions (B1)-(B3). Then the following holds.
\begin{itemize}
\item[(i)] For every $\phi\in H^1(\bbR^N),$ there exists a unique, strong $H^1$-solution $u$ of (\ref{eq:NANLSE}), which is defined on a maximal time interval $(-T_*,T^*),$ such that there exists a blow-up alternative, $\ie, $ if $T^*<\infty,$ $\|u(t)\|_{H^1}\rightarrow \infty$ as $t\nearrow T^*,$ and if $T_*<\infty,$   $\|u(t)\|_{H^1}\rightarrow \infty$ as $t\searrow -T_*.$ Moreover,
\begin{equation*}
u\in L^a_{loc}((-T_*,T^*), W^{1,b}(\bbR^N)),
\end{equation*}
for all admissible pairs $(a,b).$ 
\item[(ii)] The charge is conserved,
\begin{equation*}
\|u(t)\|_{L^2} = \|\phi\|_{L^2},
\end{equation*}
for all $t\in (-T_*,T^*).$ 
\item[(iii)] $u$ depends continuously on $\phi \ \ :$ If $\phi_n\stackrel{n\rightarrow\infty}{\rightarrow} \phi$ in $H^1,$ and if $u_n$ is the maximal solution of (\ref{eq:NANLSE}) corresponding to the initial condition $\phi_n,$ then $u_n\stackrel{n\rightarrow\infty}{\rightarrow} u$ in $C([-S,T], L^p(\bbR^N))$ for every compact interval $[-S,T]\subset (-T_*,T^*)$ and $p\in [2,\frac{2N}{N-2}) \ \ (p\in [2,\infty) , N=1).$ 

\end{itemize}
\end{proposition}

\begin{proof}
\noindent{\it (i).}
We set $r=\max (r_1,\cdots ,r_k)$ and consider the admissible pair $(q,r),$ which satisfies (\ref{eq:Admissible}). For {\it fixed} $M,T>0,$ we introduce the space 
\begin{align*}
\sY:= \{& u\in L^\infty ((-T,T),H^1(\bbR^N))\cap L^q((-T,T), W^{1,r}(\bbR^N)):\ \ \|u\|_{L^\infty((-T,T),H^1)}\le M, \\ &\|u\|_{L^q((-T,T),W^{1,r})} \le M\} 
\end{align*}
with distance 
\begin{equation}
\label{eq:DistanceY}
d(u,v):= \|u-v\|_{L^\infty((-T,T),L^2)} + \|u-v\|_{L^q((-T,T),L^r)}.
\end{equation}

We note that $(\sY,d)$ is a complete metric space. To see this, consider a sequence $\{u_n\}_{n\in {\mathbb N}}\subset \sY$ such that $d(u_n,u)\rightarrow 0$ as $n\rightarrow \infty.$ Then there exist two subsequences $\{u_{n_k}\}$ and $\{u_{n_{k'}}\}$ such that $u_{n_k}(t)\stackrel{k\rightarrow\infty}{\rightarrow} u(t)$ in $L^2(\bbR^N)$ and $u_{n_{k'}}(t)\stackrel{k'\rightarrow\infty}{\rightarrow}u(t)$ in $L^r(\bbR^N)$ for almost all $t\in (-T,T).$ It follows that  
\begin{align*}
\|u\|_{L^\infty((-T,T),H^1)} &\le \liminf_{k\rightarrow\infty} \|u_{n_k}\|_{L^\infty((-T,T),H^1)}\\
\|u\|_{L^q((-T,T),W^{1,r})} &\le \liminf_{k'\rightarrow\infty} \|u_{n_{k'}}\|_{L^q((-T,T), W^{1,r})}
\end{align*}
and  $u\in L^\infty((-T,T),H^1(\bbR^N))\cap L^q((-T,T), W^{1,r}(\bbR^N));$ see for example \cite{Ca1}, Chapter 2. Therefore, $u\in \mathsf{Y}.$

We divide the proof into five steps.

\noindent{\it Step 1. Existence .} Consider the selfadjoint operator $H_0:= -\Delta,$ with domain $H^2(\bbR^N),$ and let $U$ be the unitary propagator generated by $H_0,$ which is given by
\begin{equation*}
U(t,s)=e^{-iH_0(t-s)}, \ \ t,s\in \bbR.
\end{equation*}
Given $\phi\in H^1(\bbR^N),$ we introduce the mapping $\varphi$ which is given by 
\begin{equation}
\label{eq:EvolMapping}
\varphi(u)(t) := U(t,0)\phi - i \int_0^t ds ~ U(t,s)g(s,u(s)).
\end{equation}
We will show that for a suitable choice of $M$ and $T,$ $\varphi$ is a strict contraction on $\sY.$ Existence of a solution of (\ref{eq:NANLSE}) will then follow from Banach's fixed point theorem.


For $r_j$ appearing in Assumption (B2), we choose $q_j$ such that $(q_j,r_j)$ are admissible pairs, $j=1, \cdots, k.$ Applying H\"older's inequality (in space) gives 
\begin{equation*}
\|v\|_{W^{1,r_j}} \le \|v\|_{H^1}^{\frac{2(r-r_j)}{r_j(r-2)}} \|v\|_{W^{1,r}}^{\frac{r(r_j-2)}{r_j(r-2)}} ,
\end{equation*} 
and it follows by applying H\"older's inequality in time that 
\begin{equation*}
\|v\|_{L^{q_j}((-T,T),W^{1,r_j})} \le \|v\|_{L^\infty((-T,T),H^1)}^{\frac{2(r-r_j)}{r_j(r-2)}} \| \|v\|_{L^q((-T,T), W^{1,r})}^{\frac{r(r_j-2)}{r_j(r-2)}} .
\end{equation*}
Therefore, if $u\in \sY,$ then 
\begin{equation*}
\|u\|_{L^{q_j}((-T,T),W^{1,r_j})} \le M^{\frac{2(r-r_j)}{r_j(r-2)}} M^{\frac{r(r_j-2)}{r_j(r-2)}} = M, \ \ j=1,\cdots ,k,
\end{equation*}
and $u\in L^{q_j}((-T,T),W^{1,r_j}(\bbR^N)), j=1,\cdots, k.$ It follows from Assumption (B2) that $g_j\in  L^{q_j}((-T,T), W^{1,r_j'}(\bbR^N))$ such that  
\begin{equation*}
\|g_j(t,u)\|_{L^{q_j}((-T,T),W^{1,r_j'})} \le C(M)(1+T^\beta) (T^{\frac{1}{q_j}}+ M)
\end{equation*}
for $u\in \sY.$ It follows that 
\begin{equation}
\label{eq:NLUpperBd}
\|g_j(t,u)\|_{L^{q_j'}((-T,T),W^{1,r_j'})} \le C(M)(1+T^\beta) (T^{\frac{1}{q_j}}+ M)T^{\frac{q_j-q_j'}{q_jq_j'}}, \ \ j=1,\cdots k.
\end{equation}
Together with Strichartz estimates \footnote{We recall the results of Strichartz theorem, see \cite{Str1,GV3} and also \cite{KeTa1}. For every $\phi\in L^2(\bbR^N)$ and every admissible pair $(q,r),$ the function $t\rightarrow U(t,0)\phi \in L^q(\bbR,L^r(\bbR^N))\cap C(\bbR, L^2(\bbR^N)),$ such that 
\begin{equation*}
\|U(\cdot,0)\phi\|_{L^q(\bbR, L^r)} \le C\|\phi\|_{L^2}, \ \ \forall \phi \in L^2(\bbR^N),  
\end{equation*} 
where $C$ is a constant that depends on $q.$ Consider $I\subset \bbR$ such that $0\in I.$ Let $J\subset \overline{I}$ such that $0\in J,$ where $\overline{\cdot}$ denotes the closure. Let $(\gamma, \rho)$ be an admissible pair, and $f\in L^{\gamma'}(I, L^{\rho'}(\bbR^N)).$ Then, for all admissible pairs $(q,r),$ the function 
\begin{equation*}
t\rightarrow \Phi_f(t) = \int_0^t ds U(t,s) f(s) \in L^q(I,L^r(\bbR^N))\cap C(J,L^2(\bbR^N)),
\end{equation*}
such that 
\begin{equation*}
\|\Phi_f\|_{L^q(I,L^r)} \le C \|f\|_{L^{\gamma'}(I,L^{\rho'})}, 
\end{equation*}
where $C$ is a constant independent of $I$ and depends on $q$ and $\gamma$ only.},
we have
\begin{equation}
\label{eq:LCont}
\|\varphi (u)\|_{L^\infty((-T,T),H^1)} + \|\varphi(u)\|_{L^q((-T,T), W^{1,r})} \le C'\|\phi\|_{H^1} + C' C(M) (1+M)T^{\delta},
\end{equation}
for $u\in \sY$ and $T\le 1,$ where $C'$ is a positive constant independent of $T$ and $M$, and $\delta = \min_{j\in\{1,\cdots k\}}\frac{q_j-q_j'}{q_jq_j'}.$ Note that it follows from (\ref{eq:Admissible}) that $\delta>0.$ Furthermore, it follows from Strichartz theorem that $\varphi(u)\in C([-T,T],H^1(\bbR^N)).$ We choose $M$ and $T\le 1$ such that 
\begin{equation}
\label{eq:ChoiceMT}
C'\|\phi\|_{H^1}\le M/2,\ \ C' C(M) (1+M)T^\delta < M/2.
\end{equation}
For this choice of $T$ and $M,$ $\varphi (u)\in \sY$ for all $u\in \sY.$ 

It also follows from Assumption (B2) that 
\begin{equation*}
\|g_j(t,u)-g_j(t,v)\|_{L^{q_j}((-T,T),L^{r_j'})} \le C(M)(1+T^{\beta})\|u-v\|_{L^{q_j}((-T,T),L^{r_j})}, 
\end{equation*}
for $u,v\in\sY,$ and hence 
\begin{equation*}
\|g_j(t,u)-g_j(t,v)\|_{L^{q_j'}((-T,T),L^{r_j'})} \le C(M)(1+T^{\beta})T^{\frac{q_j-q_j'}{q_jq_j'}}\|u-v\|_{L^{q_j'}((-T,T),L^{r_j})}.
\end{equation*}
Applying Strichartz estimates, we have 
\begin{equation*}
\|\varphi(u) -\varphi(v)\|_{L^q((-T,T),L^r)} + \|\varphi(u) -\varphi(v)\|_{L^\infty((-T,T),L^2)} \le C' C(M)(1+T^\beta)T^\delta d(u,v),
\end{equation*}
where $d(\cdot,\cdot)$ is defined in (\ref{eq:DistanceY}) and $C',\delta$ appear in (\ref{eq:LCont}). If $M$ and $T$ satisfy (\ref{eq:ChoiceMT}) then $C' C(M)(1+T^\beta)T^\delta<\frac{1}{2},$ and hence the mapping $\varphi$ is a strict contraction on $\sY.$ By Banach's fixed point theorem, $\varphi$ has a unique fixed point $u\in\sY,$
\begin{equation}
\label{eq:DuhamelExpansion}
u(t) = U(t,0)\phi -i\int_0^t~ds U(t,s)g(s,u(s)),
\end{equation}
for almost all $t\in (-T,T).$ Furthermore, since $\varphi(u)\in C([-T,T],H^1(\bbR^N)),$ $u\in C([-T,T], H^1(\bbR^N)).$ 

\noindent{\it Step 2. Uniqueness.} We now use the Strichartz estimates to prove uniqueness of the solution on the interval $(-T,T).$ Suppose there exists $u$ and $v$ satisfying (\ref{eq:NANLSE}) on the interval $(-T,T).$ Then
\begin{equation*} 
w(t):= u(t)-v(t) = -i\int_0^t ds ~ U(t,s) [g(s,u(s))-g(s,v(s))] .
\end{equation*}
Let 
\begin{equation*}
w_j(t) := -i\int_0^t ds~ U(t,s) [g_j(s,u(s))-g_j(s,v(s))], \ \ j=1,\cdots ,k.
\end{equation*}
It follows from Strichartz theorem that 
\begin{align*}
& \|\int_0^t ds~ U(t,s)[g_j(s,u(s))-g_j(s,v(s))] \|_{L^{q_j} ((-T,T), L^{r_j})} \\ &\le C'\| g_j(s,u(s))-g_j(s,v(s))\|_{L^{q_j'}((-T,T), L^{r_j'}} .
\end{align*}
Together with Assumption (B2), this implies
\begin{align*}
&\|\int_0^t ds~ U(t,s)[g_j(s,u(s))-g_j(s,v(s))] \|_{L^{q_j} ((-T,T), L^{r_j})}  \\ &\le C'C(M)(1+T^\beta) \| u-v \|_{L^{q_j'}((-T,T), L^{r_j})} \\ 
&\le C'C(M)(1+T^\beta) T^{\frac{q_j-q_j'}{q_jq_j'}} \|u-v\|_{L^{q_j}((-T,T),L^{r_j})}, 
\end{align*}
Choose $T<1$ small enough such that $C'C(M) T^{\frac{q_j-q_j'}{q_jq_j'}}<\frac{1}{2}, j=1,\cdots , k.$ Then
\begin{equation*}
\|w_j\|_{L^{q_j}((-T,T),L^{r_j}(\bbR^n))} < \|w_j\|_{L^{q_j}((-T,T),L^{r_j}(\bbR^n))},
\end{equation*}
which implies $w_j=0$ on $(-T,T), j=1,\cdots, k.$ 

\noindent{\it Step 3. Blow-up alternative.}
We now prove the blow-up alternative by contradiction. We define
\begin{equation}
\label{eq:MaxT}
T^*:= \{ \sup_{T\in\bbR^+} T: \ \ {\mathrm a \ \ solution \ \ for \ \  (\ref{eq:NANLSE}) \ \  exists\ \  on \ \ } [0,T^*) \}.
\end{equation}
Suppose that $T^*<\infty$ and that $\exists M<\infty$ and a sequence $\{t_i\}_{i\in\bbN}\subset [0,T^*)$ such that $t_i\stackrel{i\rightarrow\infty}{\rightarrow} T^*$ and $\|u(t_i)\|< M$ for all $i\in\bbN.$ Choose $j\in\bbN$ such that $t_j+T(M)>T^*,$ where $T(M)$ is the time scale over which Steps 1 and 2 holds. Applying the above analysis starting with $u(t_j)$ implies the existence of the solution of (\ref{eq:NANLSE}) to times $t_j+T(M),$ which contradicts (\ref{eq:MaxT}). Therefore, $\|u\|_{H^1}\rightarrow \infty$ if $t\nearrow T^*.$ We also define 
\begin{equation*}
T_*:= \{\sup_{T\in\bbR^+} T: \ \ {\mathrm a \ \  solution \ \ for \ \  (\ref{eq:NANLSE}) \ \ exists \ \ on \ \ } (-T_*,0] \}.
\end{equation*}
Using a similar argument, one can show that $\|u\|_{H^1}\rightarrow\infty$ as $t\searrow -T_*.$ 
Therefore, (\ref{eq:NANLSE}) has a blow-up alternative.  

Note that it follows from (\ref{eq:DuhamelExpansion}), (\ref{eq:NLUpperBd}) and Strichartz theorem that 
\begin{equation*}
u\in L^a_{loc}((-T_*,T^*), W^{1,b}(\bbR^N)),
\end{equation*}
for all admissible pairs $(a,b).$ 

\noindent{\it Step 4. (ii) Charge conservation.} To prove charge conservation, we use Assumption (B3) and the fact that $u\in H^1$. Using (\ref{eq:NANLSE}), we have
\begin{equation*}
\partial_t\frac{1}{2} \|u\|^2_{L^2} = \langle iu, i\partial_t u \rangle = \langle iu, -\Delta u \rangle + \langle iu, g(t,u) \rangle =0,
\end{equation*}
and hence 
\begin{equation}
\label{eq:ChargeConservation}
\|u(t)\|_{L^2} = \|\phi\|_{L^2}, \  \ t\in (-T_*,T^*).
\end{equation}

\noindent{\it Step 5. (iii) Continuous dependence.} Consider the sequence $\{\phi_n\}_{n\in\bbN}$ such that $\phi_n\rightarrow \phi$ in $H^1$ as $n\rightarrow \infty,$ and let $u_n$ be the maximal solution of (\ref{eq:NANLSE}) corresponding to the initial condition $\phi_n.$ We claim that there exists $T>0,$ which depends on $\|\phi\|_{H^1}$ only, such that $u_n$ is well defined on $[-T,T]$ for $n$ large enough, and $u_n\rightarrow u$ in $C([-T,T],L^p(\bbR^N))$ as $n\rightarrow\infty,$ $p\in [0,\frac{2N}{N-2}), (p\in [0,\infty) \ \ N=1).$ Claim {\it (iii)} follows by repeating this property to cover any compact subset of $(-T_*,T^*).$ 

Since $\phi_n\rightarrow \phi$ in $H^1$ as $n\rightarrow \infty,$ it follows that there exists $n_0\in\bbN$ such that $\|\phi_n\|_{H^1}\le 2\|\phi\|_{H^1}, \ \ \forall n\ge n_0.$ It follows from Steps 1 and 2 above that there exists $T\equiv T(\|\phi\|_{H^1})$ such that $u_n$ and $u$ are defined on $[-T,T],$ for $n\ge n_0,$ such that 
\begin{equation*}
\|u\|_{L^\infty((-T,T),H^1)} + \sup_{n\ge n_0}\|u_n\|_{L^\infty((-T,T),H^1)} \le C'\|\phi\|_{H^1},
\end{equation*}
for some positive constant $C'.$ Note that charge conservation {\it (ii)} and $\phi_n\rightarrow \phi$ in $H^1$ imply that 
\begin{equation}
\label{eq:L2Cont}
u_n\rightarrow u \ \ \mathrm{in} \ \ C([-T,T],L^2(\bbR^N)).
\end{equation}
Furthermore, it follows from (\ref{eq:EvolMapping}) that
\begin{equation*}
u(t)-u_n(t) = U(t,0)(\phi-\phi_n) +\varphi(u)(t) - \varphi (u_n)(t), \ \ t\in [-T,T].
\end{equation*} 
Using Strichartz estimates,
\begin{align*}
&\|u-u_n\|_{L^\infty((-T,T)L^2)} + \|u-u_n\|_{L^q ((-T,T),L^r)} \\ &\le C' \|\phi-\phi_n\|_{H^1} + C'C(M)(1+T^\beta) T^\delta \times \\ &\times (\|u-u_n\|_{L^\infty((-T,T),L^2)} + \|u-u_n\|_{L^q((-T,T),L^r)}),
\end{align*}
where $\delta=\min_{j\in\{1,\cdots , k\}}\frac{q_j-q_j'}{q_jq_j'}>0.$ Choosing $T$ small enough (yet still depending on $\|\phi\|_{H^1}$) such that $C'C(M)(1+T^\beta) T^\delta <\frac{1}{2},$ 
\begin{equation}
\label{eq:YCont}
\|u-u_n\|_{L^\infty((-T,T),L^2)} + \|u-u_n\|_{L^q ((-T,T),L^r)} \le 2 C' \|\phi-\phi_n\|_{H^1}.
\end{equation}
It follows from (\ref{eq:L2Cont}), (\ref{eq:YCont}) and the Gagliardo-Nirenberg inequality that
\begin{equation*}
u_n\rightarrow u \ \ \mathrm{in} \ \ C([-T,T],L^p(\bbR^N)),
\end{equation*}
for all $p\in [2,\frac{2N}{N-2})$ ($p\in [0,\infty)$ if $N=1$). This completes the proof of the proposition.

\end{proof}

We define the energy functional
\begin{equation}
\label{eq:Energy}
E(t,u):= \frac{1}{2}\int |\nabla u|^2 dx +G(t,u),
\end{equation}
for $u\in H^1(\bbR^N).$ Note that since the nonlinearity and the potential depend on time, the energy is not conserved. We have the following proposition.

\begin{proposition}\label{pr:EnergyUpperBd}
Suppose that Assumptions (B1)-(B4) are satisfied, and let $u$ denote the solution of (\ref{eq:NANLSE}) given by Proposition \ref{pr:LWP}. Then 
\begin{equation}
\label{eq:EnergyUpperBd}
|E(t,u(t))| \le |E(0,\phi)| + T \tilde{C}(\|\phi\|_{L^2}), 
\end{equation}
for all $t\in [-T,T],$ where $[-T,T]$ is a compact subset of $(-T_*,T^*),$ and $\tilde{C}(\|\phi\|_{L^2})$ appears in Assumption (B4). 
\end{proposition}

\begin{proof}
Since Assumptions (B1)-(B3) are satisfied, the results of Proposition \ref{pr:LWP} hold. We choose a finite $T>0$ such that $T<\min(T_*,T^*).$ We know from Proposition \ref{pr:LWP} that 
\begin{equation*}
u\in L^a_{loc}((-T_*,T^*), W^{1,b}(\bbR^N)),
\end{equation*}
for all admissible pairs $(a,b).$ In particular, 
\begin{equation*}
u\in L^{q_j}((-T,T), W^{1,r_j}(\bbR^N)), \ \ j=1,\cdots ,k,
\end{equation*}
where the admissible pairs $(q_j,r_j)$ appear in Assumption (B2). We note that by Mihlin's multiplier theorem, $W^{s,p}=H^{s,p}$ for $1< p <\infty$ and $s$ an integer, see for example \cite{BerLoef1}. 

Since $\nabla$ commutes with the unitary propagator $U$ corresponding to the free time evolution, and since the $L^2$ norm is invariant under unitary transformations, we have, using the Duhamel expansion of $u$ given in (\ref{eq:DuhamelExpansion}),
\begin{align*}
&\|\nabla u(t)\|^2_{L^2} = \|\nabla U(0,t) u(t) \|^2_{L^2} \\
&= \|\nabla \phi - i \int_0^t ds~ U(0,s) \nabla g(s,u(s))\|^2_{L^2} \\
&= \|\nabla \phi\|_{L^2}^2 -2 \Im \langle \nabla \phi ,\int_0^tds~ U(0,s) \nabla g(s,u(s)) \rangle + \| \int_0^t ds~ U(0,s)\nabla g(s,u(s)) \|_{L^2}^2 \\
&= \|\nabla \phi\|_{L^2}^2 +2\Im \int_0^t ds~ \langle \nabla g(s,u(s)), U(s,0)\nabla \phi \rangle + \| \int_0^t ds~ U(0,s)\nabla g(s,u(s)) \|_{L^2}^2 .
\end{align*} 
Notice that
\begin{align*}
\| \int_0^t ds~ U(0,s)\nabla g(s,u(s)) \|_{L^2}^2 &= 2 \Re \int_0^t ds \langle \nabla g(s,u(s)), \int_0^s ds'~ U(s,s')\nabla g(s',u(s')) \rangle  \\
&= 2 \Im \int_0^tds~ \langle \nabla g(s,u(s)) , -i \int_0^s ds'~ U(s,s') \nabla g(s',u(s'))\rangle, 
\end{align*}
and hence 
\begin{align*}
\|\nabla u\|_{L^2}^2 &= \|\nabla\phi\|_{L^2}^2 + 2 \sum_{j=1}^k \Im \int_0^t ds~\langle \nabla g_j(s,u(s)) , \nabla u(s)\rangle \\
&= \|\nabla \phi\|_{L^2}^2 - 2 \sum_{j=1}^k \Im \int_0^t ds~ \langle g_j(s,u(s)), \Delta u(s)\rangle,
\end{align*}
where the scalar product is well-defined using  Assumption (B2) and duality on 
\begin{equation*}
(L^1((-T,T),H^1)+L^{q'_j}((-T,T),H^{1,r_j'}))\times (L^\infty((-T,T),H^1)\cap L^{q_j}((-T,T),H^{1,r_j}))), 
\end{equation*}
$j=1,\cdots ,k,$ see for example \cite{BerLoef1}. Now,
\begin{align*}
\Im\langle g(t,u(t)) , \Delta u(t) \rangle &= \lim_{\epsilon\searrow 0} \Im \langle (1-\epsilon\Delta)^{-1} g(t,u(t)) , (1-\epsilon \Delta)^{-1} \Delta u(t)\rangle \\
&= \lim_{\epsilon\searrow 0} \Im \langle (1-\epsilon\Delta)^{-1} g(t,u(t)), (1-\epsilon\Delta)^{-1} (-i\partial_t u(t) + g(t,u(t))) \rangle \\
&= \lim_{\epsilon\searrow 0} \Im \langle (1-\epsilon\Delta)^{-1}g(t,u(t)), -i (1-\epsilon\Delta)^{-1}\partial_t u(t) \rangle \\
&= \Re\langle g(t,u(t)), \partial_t u(t)\rangle \\
&=  \frac{d}{dt} G(t,u(t)) - (\partial_t G)(t,u(t))
\end{align*}
for almost all $t\in (-T,T),$where $G$ appears in Assumption (B4). Therefore,
\begin{equation}
\label{eq:RateH1}
\|\nabla u(t)\|_{L^2}^2 =\|\nabla \phi\|_{L^2}^2 - 2 G(t,u(t))+2 G(0,\phi) +2\int_0^t ds~ (\partial_sG)(s,u(s)). 
\end{equation}
Together with (\ref{eq:Energy}), Assumption (B4), and conservation of charge (\ref{eq:ChargeConservation}), this implies
\begin{equation*}
|E(t,u(t))| \le |E(0,\phi)| + T \tilde{C}(\|\phi\|_{L^2}), 
\end{equation*}
for $t\in [-T,T].$ The claim of the proposition follows by iterated application of this result to cover every compact subset $[-T,T]\subset (-T_*,T^*).$

\end{proof}

This is the main result of this section.

\begin{theorem}\label{th:GWP}
Suppose (B1)-(B5) hold. Then the solution $u$ of (\ref{eq:NANLSE}) with initial condition $\phi\in H^1 (\bbR^N )$ is global in $H^1,$ $\ie , T^*=T_*=\infty,$ where $T^*,T_*$ appear in Proposition \ref{pr:LWP}. 
\end{theorem}

\begin{proof}
Since Propositions \ref{pr:LWP} and \ref{pr:EnergyUpperBd} follow from Assumptions (B1)-(B4), we only need to show that $\|u(t)\|_{H^1}, \ \ t\in [0,T^*)$ is finite if $T^*<\infty,$ which, together with the blow-up alternative, implies a contradiction. Suppose $T^*<\infty.$ 
Assumption (B5), (\ref{eq:ChargeConservation}) and (\ref{eq:EnergyUpperBd}) imply that
\begin{align*}
\frac{1}{2}\|u(s)\|^2_{H^1} &= \frac{1}{2}(\|u(s)\|^2_{L^2}+\|\nabla u(s)\|^2_{L^2})\\
& \le \frac{1}{2}\|u(s)\|^2_{L^2} + |E(s,u(s))| + |G(s,u(s))| \\
& < \frac{1}{2} \|\phi\|^2_{L^2} + |E(0,\phi)| + T^*\tilde{C}(\|\phi\|_{L^2}) + \frac{1-\epsilon}{2}\|u(s)\|^2_{H^1} + C(\|\phi\|_{L^2}) , 
\end{align*}
for all $s\in [0,T^*),$
and hence 
\begin{equation*}
\|u(s)\|_{H^1} <\infty,
\end{equation*}
for finite $T^*,$ which contradicts the blow-up alternative. The case of $T_*$ is proven similarly. 

\end{proof}

\begin{remark}
One can directly verify that Assumptions (A1) and (A2) in Subsection \ref{sec:Main} imply that Assumptions (B1)-(B5) are satisfied. 
\end{remark}

\begin{remark}
Assumptions (B1)-(B5) are satisfied if  
\begin{equation}
\label{eq:ExNL}
g(t,u)(\cdot)= V(t,\cdot)u(\cdot) + f(t,\cdot, u (\cdot)) + (W(t)\star|u|^2) (\cdot) u(\cdot), \ \ u\in H^1,
\end{equation}
where $V,f$ and $W$ satisfy the following; see \cite{Ca1}. 

The potential $V$ is real valued such that $V,\nabla V\in C(\bbR,L^p(\bbR^N)+L^\infty(\bbR^N)),$ with $p>\frac{N}{2}, p\ge 1,$
and $\partial_t V \in L^\infty (\bbR, L^\infty (\bbR^N)).$ Let $g_1(t,u)(\cdot)=V(t,\cdot)u(\cdot ).$ Assumption (B2) follows by Sobolev's embedding theorem with $r_1=\frac{2p}{p-1}, \beta=1,$ and Assumption (B3) follows trivially since $V$ is real valued.

The local nonlinearity $f:\bbR\times\bbR^N\times [0,\infty)\rightarrow \bbR$ such that $f(t,x,u)$ is continuous in $t,$ measurable in $x$ and continuous in $u,$ and $f(t,x,0)=0 \forall t\in \bbR$ and almost every $x\in\bbR^N.$ If $N\ge 2, \exists $ a positive constant $C$ and $\alpha \in [0,\frac{4}{N-2})$ such that 
\begin{equation*}
|f(t,x,u)-f(t,x,v)| \le C (1+|u|^\alpha+|v|^\alpha)|u-v|, 
\end{equation*}
uniformly in $t\in \bbR,$ for almost all $x\in \bbR^N, u,v\in \bbR.$ If $N=1,$ then for every $M>0,\exists L(M)$ such that 
\begin{equation*}
|f(t,x,u)-f(t,x,v)| \le L(M)|u-v|,
\end{equation*}
uniformly in $t\in \bbR,$ for almost all $x\in \bbR^N, u,v\in \bbR,
|u|+|v|\le M.$ The local nonlinearity $f$ is extended to $\bbR\times\bbR^N\times \bbC$ by defining
\begin{equation*}
f(t,x,z):=\frac{z}{|z|} f(t,x,|z|), \ \ \forall z\in \bbC\backslash \{ 0 \},
\end{equation*}
$t\in \bbR$ and for almost all $x\in\bbR^N.$ We define 
\begin{equation}
\label{eq:g2}
g_2(t,u)(\cdot):= f(t,\cdot,u), \ \ t\in\bbR, u\in H^1,
\end{equation}
and
\begin{equation*}
G_2(t,u)=\int dx \int_0^{|u|}dr f(t,x,r), \ \ u\in H^1, \ \ t\in\bbR,
\end{equation*}
Assumption (B2)  follows from Sobolev's embedding theorem $H^1\hookrightarrow L^{\alpha+2}(\bbR^N), r_2=\alpha+2.$ Assumption (B3) follows from (\ref{eq:g2}).

For the nonlocal (Hartree type) nonlinearity $(W\star |u|^2)u,$ $W$ is a real valued function, $W:\bbR\times\bbR^N\rightarrow \bbR,$ such that $W\in L^q(\bbR^N),$ $q>\frac{N}{4}, q\ge 1.$ We make the identification $g_3(t,u)=(W(t)\star |u|^2)u, \ \ t\in \bbR, \ \ u\in H^1,$ and $G_3(t,u)=\frac{1}{4}\int dx (W\star |u|^2)|u|^2.$ Assumption (B2) follows from Young's and Sobolev's inequalities, with $r_3=\frac{4q}{2q-1}.$ 
Note that the analysis can be directly extended to the case when $f$ is a finite sum of local and nonlocal nonlinearities. 

Assumptions (B4) and (B5) follow by the application of H\"older's and the Gagliardo-Nirenberg inequality if the above holds together with the following. There exists $\delta\in [0,\frac{4}{N})$ and $C>0$ constant such that  
\begin{equation*}
\int_0^{|u|} dr f(t,x,r) \le C |u|^2 (1+|u|^\delta) ,
\end{equation*}
uniformly in $t\in\bbR,$ and 
\begin{equation*}
\partial_t\int_0^{|u|} dr f(t,x,r) \le A(t)|u|^2,
\end{equation*} 
for $t\in \bbR,$ almost all $x\in\bbR^N, $ and $u\in H^1,$ where $A(t)\in L^{\infty}(\bbR).$ Moreover, $W$ in the nonlocal nonlinearity is spherically symmetric such that
\begin{equation*}
W^+:= \max(0,W)\in L^\infty(\bbR, L^\sigma(\bbR^N)+L^\infty(\bbR^N)),
\end{equation*}
$\sigma> \frac{N}{2},\sigma\ge 1,$ and 
\begin{equation*} 
\partial_t W^+\in L^{\infty}(\bbR,L^\infty(\bbR^N)).
\end{equation*}

\end{remark}

\begin{remark}
There are several ways in which one may relax condition (B4). The above result also holds for $l(t)$ appearing in (B4) in $L^p(\bbR), \ \ 1\le p\le \infty.$ If  $l(t)$ appearing in assumption (B4) is in $L^1(\bbR),$ then one can easily show that the $H^1$ norm of the solution of the nonautonomous nonlinear Schr\"odinger equation (\ref{eq:NANLSE}) is finite for all times $t\in \bbR.$ In this case, one may extend the above analysis to potentials which grow at infinity, say quadratically; see \cite{Oh1} for a discussion in the time-independent case. 

On the other hand, if $l(t)\in L^p_{loc}(\bbR), \ \ 1\le p\le \infty,$  the upper bound on the energy functional, (\ref{eq:EnergyUpperBd}) in Proposition \ref{pr:EnergyUpperBd} is replaced with 
\begin{equation*}
|E(t,u(t))| \le |E(0,\phi)| + T C'(T) \tilde{C}(\|\phi\|_{L^2}), 
\end{equation*}
where $C'(T)>0$ is finite for finite $T.$

Furthermore, if one replaces Assumption (B4) with 
\begin{equation*}
|\partial_t G(t,u)|\le \|u\|_{L^p}l(t),
\end{equation*}
such that $l(t)\in L_{loc}^{q'}(\bbR)$ and $(q,p)$ form an admissible pair, then, applying H\"older's inequality in time and using 
\begin{equation*}
u\in L^q_{loc}((-T_*,T^*), W^{1,p}(\bbR^N)),
\end{equation*}
we have 
\begin{equation*}
|E(t,u(t))| \le |E(0,\phi)| + C(T), 
\end{equation*}
where $C(T)>0$ is finite for finite $T.$ Global well-posedness follows like before.

\end{remark}

\begin{remark}
There are relatively few rigorous results on nonautonomous nonlinear Sch\"odinger equations. The Cauchy problem for a nonautonomous nonlinear Schr\"odinger equation that is obtained by applying a pseudo-conformal transformation to a nonlinear Schr\"odinger equation with local nonlinearity is studied in \cite{CaWe1}. Moreover, well-posedness of a nonlinear Schr\"odinger equation with a local nonlinearity whose coefficient is time-periodic was studied in \cite{CKP1}, where the instability of the ground state due to periodic modulation of the nonlinearity is investigated. Furthermore, (endpoint) Strichartz estimates were obtained in \cite{SR1} for time-dependent potentials which are {\it small} and concentrated in frequency space. Time-dependent potentials also arise in the analysis of charge transfer models , \cite{RSS1}, and scattering of multisolitons, \cite{RSS2}, where Strichartz estimates were applied in order to study the asymptotic stability of multisolitons. We also mention in the linear case the dispersive estimates in \cite{Ya1} for time dependent potentials that decay in time, and the analysis in \cite{Bo1} on the slow growth of Sobolev norms for the linear Schr\"odinger equation with (quasi-) periodic potentials.

\end{remark}



\section{Properties of the nonlinear Sch\"odinger equation}\label{sec:NLSeq}

In this section, we recall some properties of the nonlinear Schr\"odinger equation (\ref{eq:NLSE}) and the soliton manifold, see for example \cite{Ca1}. We will use these properties in the following sections. 

\subsection{Symplectic, Hamiltonian and Variational structure}

The space $\Hone = H^1(\bbR^N,\bbR^2)$ as a real space, and it has a real inner product (Riemannian metric) 
\begin{equation}
\langle u,v \rangle :=\Re \int dx ~ u\overline{v} , 
\end{equation}
for $u,v\in H^1(\bbR^N,\bbR^2).$ \footnote{The tangent space $TH^1=H^1.$} It is equipped with a symplectic form
\begin{equation}
\omega(u,v):= \Im \int  dx ~ u\overline{v} = \langle u,iv\rangle.
\end{equation}
The Hamiltonian functional corresponding to the nonlinear Schr\"odinger equation (\ref{eq:NLSE}) is
\begin{equation}
\label{eq:NLHamiltonian}
H_V(\psi):=\frac{1}{2}\int (|\nabla\psi|^2 + V|\psi|^2)dx - F(\psi).
\end{equation}
Using the correspondence 
\begin{align*}
\Hone &\longleftrightarrow H^1(\bbR^N,\bbR) \oplus H^1(\bbR^N,\bbR)\\
\psi &\longleftrightarrow (\Re\psi,\Im\psi)\\
i^{-1}&\longleftrightarrow J,
\end{align*}
where $J:=\begin{pmatrix} 0 & 1 \\ -1 & 0 \end{pmatrix}$ is the complex structure on $H^1(\bbR^N,\bbR^{2}),$ the nonlinear Schr\"odinger equation can be written as 
\begin{equation*}
\partial_t \psi = J H_V' (\psi).
\end{equation*}
In the following, we denote $H^1(\bbR^N)$ as either $H^1(\bbR^N,\bbC)$ or $H^1(\bbR^N,\bbR^2).$
We note that since the external potential is time-dependent, the Hamiltonian functional $H_V$ defined in (\ref{eq:NLHamiltonian}) is nonautonomous, and there is no conservation of energy. Still, $H_V$ is invariant under global gauge transformations, 
\begin{equation*}
H_V(e^{i\gamma}\psi)=H_V(\psi),
\end{equation*}
and the associated conserved Noether charge is the ``mass'' 
\begin{equation}
N(\psi):= \frac{1}{2}\int dx~ |\psi|^2.
\end{equation}

Orbital stability (Assumption (A5), Subsection \ref{sec:Assumptions}) implies that $\eta_\mu$ appearing in Assumption (A4) is a local minimizer of $H_{V=0}(\psi)$ restricted to the balls ${\mathcal B}_m:= \{\psi\in H^1 : N(\psi)=m\},$ for $m>0;$ see \cite{GSS1,GSS2}. They are critical points of the functional 
\begin{equation}
\label{eq:EnergyFunctional}
\cE_\mu (\psi) := \frac{1}{2} \int dx~ (|\nabla\psi|^2 +\mu|\psi|^2)-F(\psi),
\end{equation}
where $\mu=\mu(m)$ is a Lagrange multiplier.

\subsection{Soliton Manifold}

When $V=0,$ Assumption (A3) implies that the nonlinear Schr\"odinger equation (\ref{eq:NLSE}) is invariant under spatial translations, time translations, gauge transformations, spatial rotations and Galilean transformations. The corresponding conserved quantities are the field momentum, energy, mass, angular momentum, and {\it center of mass motion},
\begin{equation*}
\int \bar{\psi}(-i\nabla)\psi ,  \ \ \frac{1}{2}\int |\nabla\psi|^2 -F(\psi),  \ \ \int |\psi|^2 , \ \ \int \overline{\psi} (x\wedge-i\nabla)\psi , \ \ 
\int \bar{\psi} (x+2\i t \nabla)\psi .
\end{equation*}

When $V\ne 0,$ the above quantities are generally no more conserved. In particular, the rate of change of energy is
\begin{equation}
\label{eq:RateEnergy1}
\partial_t H_V(\psi) = \frac{1}{2} \int dx~ \partial_tV|\psi|^2,
\end{equation}
and the rate of change of momentum is 
\begin{equation}
\label{eq:Ehrenfest1}
\partial_t \langle\psi,-i \nabla \psi\rangle = - \langle\psi,(\nabla V)\psi\rangle.
\end{equation}
Formally, Eq. (\ref{eq:RateEnergy1}) follows from (\ref{eq:NLSE}) and (\ref{eq:NLHamiltonian}), while (\ref{eq:Ehrenfest1}), which is a statement of Ehrenfest's Theorem, follows from (\ref{eq:NLSE}). We refer the reader to Appendix \ref{sec:Appendix} for a proof of (\ref{eq:RateEnergy1}) and (\ref{eq:Ehrenfest1}).

We introduce the combined transformation $T_{av\gamma}.$
\begin{equation}
\label{eq:CombTrans}
\psi_{av\gamma}:= T_{av\gamma}\psi = e^{i(\frac{1}{2}v\cdot (x-a)+\gamma)}\psi(x-a),
\end{equation}
where $v,a\in\bbR^N$ and $\gamma\in [0,2\pi).$ We define the soliton manifold as 
\begin{equation}
\label{eq:SolMan}
\cM_s:= \{T_{va\gamma}\eta_\mu : a,v,\gamma,\mu \in \bbR^N\times\bbR^N\times [0,2\pi)\times I \},
\end{equation}
where $I$ appears in Assumption (A4). If $f'(0)=0,$ where $f$ appears in (\ref{eq:NLSE}), then $I\subset \bbR^+.$ The tangent space to the soliton manifold $\cM_s$ at $\eta_\mu\in \cM_s$ is given by 
\begin{equation}
\label{eq:TangentSpace}
\cT_{\eta_\mu}\cM_s = span\{ e_t,e_g,e_b,e_s\},
\end{equation}
where
\begin{align*}
e_t &:= \nabla_a T_a^{tr}\eta_\mu|_{a=0}=-\nabla\eta_\mu\\
e_g &:= \partial_\gamma T_\gamma^g \eta_\mu|_{\gamma=0} = i\eta_\mu\\
e_b &:= 2\nabla_v T_v^{gal}\eta_\mu|_{v=0}=ix\eta_\mu\\
e_s &:= \partial_\mu\eta_\mu.
\end{align*}

\begin{remark}
In the case of pure local nonlinearities, $f(\psi)=\lambda|\psi|^{s}\psi,$ and $V=0,$ the nonlinear Schr\"odinger equation (\ref{eq:NLSE}) is invariant under scaling, 
\begin{equation*}
T_\mu^s : \psi(x,t)\rightarrow \mu^{\frac{1}{s}}\psi(\sqrt{\mu}x,\mu t).
\end{equation*}
 In this case, one can define the generalized transformation $\overline{T}_{av\gamma\mu}:=T_{av\gamma} T_\mu^s.$ 
Furthermore, $e_\mu = \partial_\mu T_\mu^s \eta|_{\mu=1}=\frac{1}{2}(\frac{2}{s}+x\cdot\nabla)\eta_1.$
\end{remark}

\begin{remark}\label{rm:ZeroModes}
Note that the solitary wave solution $\eta_\mu$ breaks the translation and gauge symmetries of the nonlinear Schr\"odinger equation, which leads to associated zero modes of the {\it Hessian} $\cL_\mu,$ defined in Assumption (A6), Subsection \ref{sec:Assumptions}. Differentiating ${\mathcal E}_\mu' (T_a^{tr} T_\gamma^g \eta_\mu)=0$ with respect to $a$ and $\gamma$ and setting the latter two to zero gives 
\begin{equation*}
\cL_\mu e_t = 0 \; , \ \ \ \  \cL_\mu e_g = 0 \; , 
\end{equation*}
while a direct computation gives
\begin{equation*}
\cL_\mu e_b = 2i e_t \; , \;  \cL_\mu e_s = i e_g .
\end{equation*}
\end{remark}

\begin{remark}
The soliton manifold $\cM_s$ inherits a symplectic structure from $(H^1,\omega).$ For $\sigma = \{a,v,\gamma,\mu\} \in \bbR^N\times\bbR^N\times [0,2\pi)\times I,$
\begin{equation*}
\Omega_\sigma^{-1} := J^{-1}|_{\cT_{\eta_\sigma}} = P_\sigma J^{-1} P_\sigma ,
\end{equation*}
where $P_\sigma$ is the $L^2$ orthogonal projection onto $\cT_{\eta_\sigma}\cM_s.$ It can be shown that $\Omega_\sigma^{-1}$ is invertible if $\partial_\mu m(\mu)>0,$ where the mass $m(\mu)=\frac{1}{2}\int dx ~\eta_\mu^2 ,$ see \cite{FJGS1}. Explicitly,
\begin{equation}
\label{eq:Metric}
\Omega_{\mu}^{-1}|_{\{ e_k \}} := (\langle e_j , J^{-1} e_k\rangle)_{1\le j,k\le 2N+2} = 
\left(
\begin{matrix}
0 & -m(\mu) {\mathbf 1} & 0 & 0 \\
m(\mu) {\mathbf 1} & 0 & 0 & 0 \\
0 & 0 & 0 & m'(\mu) \\
0 & 0 & -m'(\mu) & 0
\end{matrix}
\right),
\end{equation}
where $e_k, k=1,\cdots , 2N+2,$ are basis vectors of $\cT_{\eta_\mu} \cM_s$ and ${\mathbf 1}$ is the $N\times N$ identity matrix. Note that $\Omega_\sigma^{-1}$ is related to $\Omega_\mu^{-1}$ by a similarity transformation.

\end{remark}

\begin{remark}\label{rm:GroupStructure}
It has been noticed in \cite{HZ1} that in one dimension, the action of the combined transformation on elements of the soliton manifold has a group structure, the Heisenberg group $\mathsf{H}^3.$ In $N$-dimensions, the group corresponds to ${\mathsf H}^{2N+1},$ the Heissenberg group in $2N+1,$ which is given by
\begin{equation*}
(a,v,\gamma)\cdot(a',v',\gamma')=(a'',v'',\gamma''),
\end{equation*}
where $a''=a+a',$ $v''=v+v',$ $,\gamma''=\gamma'+\gamma+\frac{1}{2}va'.$ Note that the Heisenberg group is a central extension of the additive group.

\end{remark}

\subsection{Skew-Orthogonal Decomposition}\label{sec:SOD}

Consider the manifold $\cM_s' = \{\eta_\sigma, \ \ \sigma\in \Sigma_0\}, \ \ \Sigma_0 = \bbR^N\times \bbR^N \times [0,2\pi) \times I_0,$ where $I_0\subset I\backslash \partial I$ is bounded. We define the $\delta$ neighbourhood of $\cM_s'$ in $H^1$ as 
\begin{equation*}
U_\delta := \{ \psi\in H^1 , \ \ \inf_{\sigma\in \Sigma_0} \|\psi - \eta_\sigma\|\le \delta\}.
\end{equation*}
Then, for $\delta$ small enough and for all $\psi\in U_\delta,$ there exists a unique $\sigma(\psi)\in C^1(U_\delta,\Sigma)$ such that 
\begin{equation*}
\omega(\psi-\eta_{\sigma(\psi)}, e)= \langle \psi - \eta_{\sigma(\psi)}, J^{-1}e\rangle =0,
\end{equation*}
for all $e\in\cT_{\eta_{\sigma(\psi)}}\cM_s.$ For a proof of this statement, we refer the reader to \cite{FJGS1}, see also \cite{Ar1}.


\section{Proof of the main result}\label{sec:ProofMain}

The proof of the main result is based on an extension of the analysis in \cite{FJGS1}, except that one needs to keep track of additional terms due to the time-dependence of the potential. Formally, the proof boils down to decomposing the solution into a component which belongs to the soliton manifold plus a fluctuation which is skew-orthogonal to the soliton manifold. The dynamics of the component belonging to the soliton manifold is effectively determined by the restriction of the Hamiltonian flow generated by the nonlinear Schr\"odinger equation to the soliton manifold, while the $H^1$ norm of the fluctuation is controlled using an approximate Lyaponuv functional. An additional feature of our analysis is an iteration scheme which gives a much longer time scale, $O(|\log h|/h),$ over which one can control the fluctuation, compared to $O(h^{-1})$ in \cite{FJGS1}. 

\subsection{Reparametrized equations of motion}\label{sec:RepEqMotion}

Suppose $\psi$ satisfies the initial value problem (\ref{eq:NLSE}) such that $\psi\in U_\delta \subset H^1,$ where $U_\delta$ appears in Subsection \ref{sec:SOD}. By the skew-orthogonal decomposition, there exists a unique $\sigma=\sigma(\psi)=\{a,v,\gamma,\mu\}\in \Sigma = \bbR^N\times \bbR^N\times [0,2\pi)\times I$ and $w'\in H^1$ such that 
\begin{align*}
\psi &= \eta_\sigma + w' \; , \\
w' & \perp J^{-1} \cT_{\eta_\sigma}\cM_s.   
\end{align*}  
Let 
\begin{equation}
\label{eq:CMSol}
u:= T^{-1}_{av\gamma}\psi = \eta_\mu + w,
\end{equation}
where $w=T^{-1}_{av\gamma}\eta_\mu.$ We introduce the anti-selfadjoint operators 
\begin{equation*}
L_j=\partial_{x_j}, \ \ L_{N+j}= -J x_j , \ \ j=1,\cdots ,N,
\end{equation*}
\begin{equation}
\label{eq:Generators}
L_{2N+1} = -J, \ \ L_{2N+2} = \partial_\mu ,
\end{equation}
with coefficients 
\begin{equation*}
\alpha_j = \partial_t a_j - v_j , \ \ \alpha_{N+j}=-\frac{1}{2} \partial_t v_j - \partial_{x_j}V(t,a) , \ \ j=1,\cdots ,N ,
\end{equation*}
\begin{equation}
\label{eq:Coefficients}
\alpha_{2N+1} = \mu -\frac{1}{4}v^2 + \frac{1}{2}\partial_t a \cdot v - V(t,a)-\partial_t\gamma , \ \ \alpha_{2N+2}=-\partial_t \mu.
\end{equation}
We denote by 
\begin{equation}
\label{eq:Alpha}
|\alpha|:= \sup_{i\in \{1,\cdots ,2N+2\}}|\alpha_i|,
\end{equation}
and $C(\alpha,w,h):= |\alpha|\|w\|_{H^1}+h^2 + \|w\|_{H^1}^2.$ We have the following proposition.

\begin{proposition}\label{pr:RepEqMotion}
Suppose that $\psi\in H^1$ satisfy (\ref{eq:NLSE}) such that $\psi(t)\in U_\delta$ for $t\in [0,T].$ Then the parameter $\sigma=\{a,v,\gamma,\mu\}$ and $w\in H^1,$ as given above, satisfy 
\begin{align*}
\partial_t a_j &= v_j + O(C(\alpha,w,h))\\
\partial_t v_j &= -2 \partial_{x_j} V(t,a) + O(C(\alpha ,w,h)) \\
\partial_t \gamma &= \mu -\frac{1}{4} v^2 +\frac{1}{2}\partial_t a \cdot v -V(t,a) + O(C(\alpha, w,h)) \\
\partial_t \mu &= O(C(\alpha,w,h)) ,
\end{align*}
for $t\in (0,T), j=1,\cdots N.$
\end{proposition}

\begin{proof}
We proceed in three steps.

\noindent {\it Step 1. Equations of motion in center of mass reference frame.} 

We first find the equation of motion for $u=T^{-1}_{av\gamma}\psi = e^{-\frac{i}{2}(v\cdot x + \gamma)} \psi(x+a).$ Differentiating $u$ with respect to $t$ and using (\ref{eq:NLSE}), and the fact that 
\begin{align*}
e^{-\frac{i}{2}(v\cdot x +\gamma)}\Delta \psi(x+a) &= \Delta u + i v\cdot \nabla u -\frac{v^2}{4} \\
e^{-\frac{i}{2} (v\cdot x + \gamma)}f(\psi(x+a)) &= f(u),
\end{align*}
we have 
\begin{equation}
\label{eq:CenterMassDyn}
\partial_t u = -J ((-\Delta +\mu)u -f(u)) + \sum_{j=1}^{2N+1}\alpha_j L_j u + J \cR_V u
\end{equation}
where 
\begin{equation*}
\cR_V := V(t,x+a) - V(t,a)- \nabla V(t,a) \cdot x,
\end{equation*}
and $L_j$ and $\alpha_j$ are as defined in (\ref{eq:Generators}) and (\ref{eq:Coefficients}) respectively. 
In other words, 
\begin{equation*}
\partial_t u = J \cE_\mu'(u) + \sum_{j=1}^{2N+1} \alpha_j L_j u +J\cR_V u,
\end{equation*}
where $\cE$ appears in (\ref{eq:EnergyFunctional}). 

\noindent{\it Step 2. Reparametrized equations of motion.} 

We now use (\ref{eq:CenterMassDyn}) and the skew-orthogonal decomposition to find equations for the parameters $\sigma=\{a,v,\gamma,\mu\}$ and $w.$ Recall that 
\begin{equation*}
\cE_\mu'(\eta_\mu)=0,
\end{equation*}
which implies 
\begin{equation*}
\cE_\mu' (u)= \cL_\mu w + N_\mu (w), 
\end{equation*}
where $\cL_\mu = (-\Delta +\mu -f'(\eta_\mu)) = \cE_\mu''(\eta_\mu)$ and $N_\mu (w)= f(\eta_\mu+w) - f(\eta_\mu) -f'(\eta_\mu)w.$ 
Substituting back in (\ref{eq:CenterMassDyn}) implies
\begin{equation*}
\partial_t w = (J\cL_\mu + \sum_{j=1}^{2N+1}\alpha_j L_j + J\cR_V)w + N_\mu (w) + J\sum_{j=1}^{2N+2} \alpha_j L_j \eta_\mu + J \cR_V \eta_\mu .
\end{equation*}
We know that $\langle Jz, w\rangle = 0$ for all $z\in \cT_\eta \cM_s.$ It follows that 
\begin{equation*}
\partial_t \langle Jz, w\rangle = \partial_t \mu \ \ \langle J\partial_\mu z, w\rangle + \langle Jz, \partial_t w \rangle = 0.
\end{equation*}  
Therefore,
\begin{align*}
\partial_t \mu\ \   \langle J\partial_\mu z, w\rangle &= - \langle Jz, \partial_t w\rangle \\
&= -\langle Jz, J \cL_\mu w\rangle - \langle Jz, J \cR_V(\eta_\mu + w) \rangle - \langle Jz, N_\mu (w)\rangle - \\ 
& \ \ - \langle Jz, \sum_{j=1}^{2N+1} \alpha_j L_j w\rangle - \langle Jz ,\sum_{j=1}^{2N+2} \alpha_j L_j \eta_\mu \rangle .
\end{align*}
It follows from Remark (\ref{rm:ZeroModes}), Section \ref{sec:NLSeq}, that $\langle Jz, J\cL_\mu w\rangle =0.$ Together with 
\begin{equation*}
[L_j, J]= 0, \ \ L_j^* = - L_j, \ \ j=1,\cdots ,2N+2,
\end{equation*}
this implies
\begin{equation}
\label{eq:IntDyn}
\sum_{j=1}^{2N+2} \alpha_j \langle JL_j z, w\rangle = -\langle z, \cR_V(w+\eta_\mu) + N_\mu (w) \rangle + \sum_{j=1}^{2N+2} \langle Jz, \alpha_j L_j \eta_\mu \rangle .
\end{equation}
Choosing $z=e_k,$ where $e_k, k\in \{1,\cdots , 2N+2\},$ is a basis vector of $\cT_\eta \cM_s$ gives 
\begin{equation*}
\sum_{j=1}^{2N+2} (\Omega^{-1})_{kj} \alpha_j = \langle e_k, N_\mu (w)+\cR_V(w+\eta_\mu)\rangle + \sum_{j=1}^{2N+2}\alpha_j \langle L_j e_k, Jw \rangle ,
\end{equation*} 
where $\Omega^{-1}$ appears in (\ref{eq:Metric}).
Replacing the definition of $\alpha_j,$ appearing in (\ref{eq:Coefficients}), in (\ref{eq:IntDyn}), and using the fact that 
\begin{equation*}
\langle x_j \eta_\mu, J\cR_V \eta_\mu \rangle = 0 , \ \ \eta_\mu (x) = \eta_\mu (|x|), \ \ \langle J\eta_\mu, \cR_V \eta_\mu\rangle =0,
\end{equation*}
gives
\begin{equation}
\label{eq:a}
\partial_t a_k = v_k + \frac{1}{m(\mu)}(\langle x_k\eta_\mu, J N_\mu (w)+ J\cR_Vw \rangle + \sum_{j=1}^{2N+2} \alpha_j \langle L_j x_k \eta_\mu, w\rangle) , 
\end{equation}
\begin{equation}
\label{eq:v}
\partial_t v_k = -2 \partial_{x_k} V(t,a) + \frac{2}{m(\mu)} (\langle \partial_{x_k}\eta_\mu, N_\mu (w)+ \cR_V w\rangle - \sum_{j=1}^{2N+2} \alpha_j \langle L_j \partial_{x_k}\eta_\mu, Jw  \rangle + \langle \partial_{x_k} \eta_\mu, \cR_V \eta_\mu\rangle ),
\end{equation}
\begin{align}
\partial_t \gamma = & \mu -\frac{1}{4}v^2 + \frac{1}{2}\partial_t a\cdot v - V(t,a) - \nonumber\\ 
& - \frac{1}{m'(\mu)} (\langle \partial_\mu\eta_\mu, N_\mu (w) + \cR_V(w+\eta_\mu) \rangle -\sum_{j=1}^{2N+2} \alpha_j \langle L_j \partial_\mu \eta_\mu, Jw\rangle ),\label{eq:g}
\end{align}
\begin{equation}
\partial_t \mu =\frac{1}{m'(\mu)} \langle \eta_\mu , JN_\mu (w) + J \cR_V w \rangle -\sum_{j=1}^{2N+2} \alpha_j \langle L_j \eta_\mu ,w\rangle .\label{eq:m}
\end{equation}

The claim directly follows from Assumptions (A1) and (A2), which imply that 
\begin{equation*}
\|\cR_V e_k\|_{L^2} = O(h^2 ) , \ \ \|N_\mu (w) \|_{H^1} \le C \|w\|_{H^1},
\end{equation*}
for $\|w\|_{H^1}\le 1 ;$ see Remark \ref{rm:Nonlinearity}, Subsection \ref{sec:Assumptions}.
 
\end{proof}

\subsection{Control of the fluctuation}

In this subsection, we use an approximate Lyaponuv functional to obtain an explicit control on $\|w\|_{H^1}$ and $|\alpha|.$ This approach dates back to \cite{We1,We2}, and has been used in \cite{FJGS1,FJGS2}. We define the Lyapunov functional 
\begin{equation}
\label{eq:LyapunovFunctional}
\cC_\mu (u,v):= \cE_\mu (u)- \cE_\mu (v) , \ \ u,v\in H^1(\bbR^N),
\end{equation}
where $\cE_\mu$ is defined in (\ref{eq:EnergyFunctional}). We proceed by estimating upper and lower bounds for $\cC_\mu(u,\eta_\mu),$ where $u=\eta_\mu + w$ appears in (\ref{eq:CMSol}).

\subsubsection{An upper bound for the Lyapunov functional}

\begin{lemma}\label{lm:UpperBoundLF}
Suppose $\psi$ satisfies (\ref{eq:NLSE}) such that $\psi(t)\in U_\delta, \ \ t\in [0,T],$ for some $\delta >0,$ and let $u,w,\eta_\mu$ be as defined in Subsection \ref{sec:RepEqMotion}. Then there exists a constant $c$ independent of $h$ such that 
\begin{equation}
\label{eq:UpperBoundLF}
\cC_\mu (u,\eta_\mu)\le ct (h^2 \|w\|_{H^1} + (|\alpha|+h) \|w\|^2_{H^1}),
\end{equation}
where $|\alpha|$ appears in (\ref{eq:Alpha}).

\end{lemma}

\begin{proof}
Recall that 
\begin{align*}
\cE_\mu (t,u) &= \frac{1}{2} \int dx |\nabla u|^2 + \mu |u|^2 -F(t,u) \\
&= H_V(u) + \frac{1}{2}\mu \|u\|_{L^2}^2 - \frac{1}{2} \int dx V |u|^2 \\
&= H_V(T_{av\gamma}^{-1}\psi) + \frac{1}{2}\mu \|T_{av\gamma}^{-1}\psi\|_{L^2}^2 - \frac{1}{2} \int dx V |T_{av\gamma}^{-1}\psi|^2.
\end{align*}
By translational symmetry, 
\begin{equation*}
\|u\|_{L^2}^2 = \|\psi\|_{L^2}^2 , \ \ \int dx V |u|^2 = \int V_{-a} |\psi|^2,
\end{equation*}
where $V_{-a}(x)\equiv V(x-a).$ Furthermore, 
\begin{equation*}
H_V(T_{av\gamma}^{-1}\psi) = H_V(\psi)+ \frac{1}{2}(\frac{1}{4}v^2 + \mu)\|\psi\|_{L^2}^2 - \frac{1}{2}v\cdot \langle i\psi, \nabla \psi \rangle + \frac{1}{2}\int dx (V_{-a} - V)|\psi|^2 ,
\end{equation*}
and hence 
\begin{equation}
\label{eq:CMEnergy}
\cE_\mu (u)= H_V(\psi) + \frac{1}{2}(\frac{1}{4}v^2 + \mu)\|\psi\|_{L^2}^2 - \frac{1}{2}v\cdot \langle i\psi ,\nabla \psi \rangle - \frac{1}{2}\int dx  V |\psi|
\end{equation}
We have the following relationships regarding the rate of change of field energy and momenta.
\begin{equation}
\label{eq:RatePotential}
\partial_t \frac{1}{2} \int dx  V |\psi|^2 = \langle \nabla V i\psi, \nabla \psi \rangle + \frac{1}{2} \langle \psi, \partial_t V\psi \rangle , 
\end{equation}
\begin{equation}
\label{eq:RateEnergy}
\partial_t H_V(\psi) = \frac{1}{2} \langle \psi, \partial_t V \psi\rangle
\end{equation}
and Ehrenfest's theorem
\begin{equation}
\label{eq:Ehrenfest}
\partial_t \langle i\psi, \nabla\psi\rangle = -\langle \psi ,\nabla V \psi\rangle. 
\end{equation}
Using (\ref{eq:NLSE}), the above three statements (\ref{eq:RatePotential}),(\ref{eq:RateEnergy}) and (\ref{eq:Ehrenfest}) are formally verified. Rigorously, one can verify them using a regularization scheme and a limiting procedure, and we refer the reader to the Appendix  for a proof of these statements; see also the proof of Proposition \ref{pr:EnergyUpperBd}, Section \ref{sec:Well-Posedness}. Furthermore, it follows from gauge invariance of (\ref{eq:NLSE}) that the charge is conserved,
\begin{equation}
\label{eq:ConservationCharge}
\partial_t\|\psi\|_{L^2} = 0,
\end{equation} 
see Theorem \ref{th:GWP} in Section \ref{sec:Well-Posedness}.
Differentiating (\ref{eq:CMEnergy}) with respect to $t$ and using (\ref{eq:RatePotential}-\ref{eq:ConservationCharge}) gives
\begin{align}
\partial_t \cE_\mu (u) &= \partial_t H_V(\psi) +\frac{1}{2}(\frac{\partial_t v\cdot v}{2} + \partial_t \mu)\|\psi\|_{L^2}^2 - \frac{1}{2}\partial_t v \cdot \langle i\psi, \nabla \psi\rangle + \frac{1}{2} v\cdot \langle \psi ,\nabla V \psi\rangle \nonumber \\ &- \langle \nabla V i\psi, \nabla \psi \rangle-  \nonumber - \frac{1}{2} \langle \psi, \partial_t V\psi \rangle  \nonumber \\
&= \frac{1}{2}(\frac{\partial_t v\cdot v}{2} + \partial_t \mu)\|\psi\|_{L^2}^2- \frac{1}{2}\partial_t v \cdot \langle i\psi, \nabla \psi\rangle+ \frac{1}{2} v\cdot \langle \psi ,\nabla V \psi\rangle- \langle \nabla V i\psi, \nabla \psi \rangle \nonumber \\
&= \frac{1}{2}\partial_t \mu \|u\|_{L^2}^2 - \langle \frac{1}{2} \langle i(\partial_t v + 2\nabla V_{a})u,\nabla u\rangle, \label{eq:RateEnergyFunctional} 
\end{align}
where we have used $u(x)=e^{i(\frac{1}{2}v\cdot x +\gamma)} \psi(x+a)$ and translation invariance of the integral in the last line. 
Furthermore, it follows from (\ref{eq:EnergyFunctional}) that 
\begin{equation*}
\partial_t \cE_\mu(\eta_\mu ) = \frac{1}{2} \partial_t \mu \|\eta_\mu\|^2,
\end{equation*}
which, together with (\ref{eq:LyapunovFunctional}) and (\ref{eq:RateEnergyFunctional}), implies
\begin{equation}
\label{eq:RateLyapunov}
\partial_t \cC_\mu(u,\eta_\mu) = \frac{1}{2} \partial_t \mu (\|u\|^2_{L^2}- \|\eta_\mu\|^2_{L^2}) - \langle i(\frac{1}{2}\partial_t v + \nabla V_a) u , \nabla u \rangle . 
\end{equation}

We now estimate both terms in (\ref{eq:RateLyapunov}). Since $\langle i z, w\rangle = 0$ for all $z\in \cT_\mu \cM_s,$ it follows from the skew-orthogonal decomposition (Subsection \ref{sec:SOD}) that 
\begin{equation*}
\|u\|^2_{L^2}- \|\eta_\mu\|^2_{L^2} = \|w\|^2_{L^2},
\end{equation*}
and hence
\begin{equation}
\label{eq:FTEstimate}
\frac{1}{2} \partial_t \mu (\|u\|^2_{L^2}- \|\eta_\mu\|^2_{L^2}) = O(|\alpha| \|w\|^2_{L^2}).
\end{equation}

To estimate the second term, we replace $u=\eta_\mu+w,$ and use the fact that $\langle ie_g,w\rangle=\langle i \nabla \eta_\mu, w\rangle = 0,$ and $\langle ih\eta_\mu, \nabla \eta_\mu \rangle = 0$ for all real $h\in L^\infty(\bbR^N).$ We have 
\begin{equation*}
\langle i(\frac{1}{2}\partial_t v + \nabla V_a) u , \nabla u \rangle = \langle i\nabla V_a w, \nabla \eta_\mu\rangle + \langle i \nabla V_a \eta_\mu , \nabla w\rangle + \langle i (\frac{1}{2}\partial_t v + \nabla V_a)w, \nabla w\rangle.
\end{equation*}
Adding and subtracting $\nabla V(t,a)\cdot \langle iw,\nabla \eta_\mu \rangle =  \nabla V(t,a) \cdot \langle i\eta_\mu, \nabla w \rangle = 0$ gives
\begin{align*}
\langle i(\frac{1}{2}\partial_t v + \nabla V_a) u , \nabla u \rangle = 
& (\frac{1}{2}\partial_t v + \nabla V(a)) \langle iw, \nabla w\rangle + \langle (\nabla V_a -\nabla V(a))iw,\nabla w\rangle  \\
&+ \langle (\nabla V_a -\nabla V(a))i\eta_\mu,\nabla w\rangle + \langle (\nabla V_a -\nabla V(a))iw,\nabla \eta_\mu\rangle.
\end{align*}
The first term of the above equation is of order $O(|\alpha|\|w\|^2_{H^1}),$ while Assumption (A2) implies that the second term is of order $O(h \|w\|_{H^1}^2).$ Assumptions (A2) and (A4) imply that the third and forth terms are of order $O(h^2 \|w\|_{H^1}).$ Hence the claim of the lemma.
\end{proof}

\subsubsection{A lower bound for the Lyapunov functional}

In this subsection, we estimate a lower bound for $\cC_\mu (u,\eta_\mu).$ Let 
\begin{equation*}
X_\mu := \{w\in H^1(\bbR^N) : \ \ \langle w, J^{-1} z\rangle =0, \forall z\in \cT_{\eta_\mu}\cM_s \}.
\end{equation*}
It follows from the coercivity property of $\cL_\mu$ that there exists a positive constant 
\begin{equation}
\label{eq:Coercivity}
\rho:= \inf_{w\in X_\mu}\langle w,\cL_\mu w\rangle >0.
\end{equation}
We refer the reader to the Appendix D in \cite{FJGS1} for a proof of this statement. We have the following result. 

\begin{lemma}\label{lm:LFLowerBound}
Suppose $\psi$ satisfies (\ref{eq:NLSE}) such that $\psi(t)\in U_\delta, \ \ t\in [0,T],$ for some $\delta >0,$ and let $u,w,\eta_\mu$ be as defined in Subsection \ref{sec:RepEqMotion}. Then there exists positive constants $\rho$ and $c$ independent of $h$ such that, for $\|w\|_{H^1}\le 1,$ 
\begin{equation}
\label{eq:LFLowerBound}
|\cC_\mu (u,\eta_\mu)| \ge \frac{\rho}{2} \|w\|_{H^1} - c\|w\|_{H^1}^3,
\end{equation}
where $\rho$ appears in (\ref{eq:Coercivity}).
\end{lemma}

\begin{proof}
We first expand $\cE_\mu(u)$ around $\eta_\mu,$ which is a critical point of $\cE_\mu.$
\begin{equation}
\cE_\mu (\eta_\mu + w) = \cE_\mu (\eta_\mu) + \frac{1}{2}\langle w, \cL_\mu w\rangle + R_\mu^{(3)}(w),
\end{equation}
where 
\begin{equation*}
R_\mu^{(3)}(w)= F(\eta_\mu +w) - F(\eta_\mu)-\langle F'(\eta_\mu),w\rangle -\frac{1}{2} \langle F''(\eta_\mu)w,w\rangle.
\end{equation*}
It follows from Assumption (A1) that 
\begin{equation*}
|R_\mu^{(3)}(w)| \le c \|w\|_{H^1}^3,
\end{equation*}
for $\|w\|_{H^1}\le 1,$ where $c>0$ is independent of $t\in \bbR;$ see Remark \ref{rm:Nonlinearity} in Subsection \ref{sec:Assumptions}. Furthermore, the coercivity property (\ref{eq:Coercivity}) implies 
\begin{equation*}
\langle w,\cL_\mu w\rangle \ge \rho \|w\|_{H^1}^2,
\end{equation*}
and hence
\begin{equation*}
|\cC_\mu (u,\eta_\mu)| = |\cE_\mu(u)-\cE_\mu (\eta_\mu)| \ge \frac{1}{2}\rho \|w\|_{H^1}^2 - c\|w\|_{H^1}^3,
\end{equation*}
for $\|w\|_{H^1}\le 1.$

\end{proof}

\subsubsection{Upper bound on the fluctuation}

In this subsection, we use the upper and lower bounds on the Lyapunov functional to obtain an upper bound on $\|w\|_{H^1}.$

\begin{proposition}\label{pr:UpperBdFluctuation}
Suppose (A1)-(A7) hold, and let $\psi$ satisfy (\ref{eq:NLSE}), and $u,\eta_\mu,w$ as above. For $h\ll 1,$ choose $T\in \bbR^+$ such that $\psi(t)\in U_{\delta}, \ \ t\in [0,T],$ where $U_\delta$ appears in Subsection \ref{sec:SOD}. Fix $\epsilon \in (0,1),$ and choose $t_0\in [0,T]$ such that $\|w(t_0)\|_{H^1}^2 < h^{2-\epsilon}.$ Then, for $h$ small enough, there exist absolute constants $C_1>1$ and $C_2>0,$ which are independent of $h,\epsilon$ and $t,$ such that
\begin{align*}
\sup_{t\in [t_0,t_0+\tau ]}\|w(t)\|_{H^1}^2 &\le C_1 (h^2 + \|w(t_0)\|_{H^1}^2)) \\
\sup_{\stackrel{t\in [t_0,t_0+\tau ]}{i\in\{1,\cdots ,2N+2\}}} |\alpha_i(t)| &\le C_1 (h^2 + \|w(t_0)\|_{H^1}^2), 
\end{align*}
where $\tau =C_2/h.$

\end{proposition}

Note that the conditions of Proposition \ref{pr:UpperBdFluctuation} are satisfied for $t_0=0.$ 

\begin{proof}
It follows from Lemma \ref{lm:UpperBoundLF} that, for $t\ge t_0,$
\begin{equation*}
|\cC_\mu (\eta_\mu+w(t),\eta_\mu)| \le |\cC_\mu (\eta_\mu+ w(t_0),\eta_\mu)| + c(t-t_0) (|\alpha| \|w(t)\|^2_{H^1} + h^2 \|w(t)\|_{H^1} + h \|w(t)\|^2_{H^1}).
\end{equation*}
Furthermore, expanding $\cE_\mu(\eta_\mu + w(t_0))$ around $\eta_\mu$ and using Assumption (A1) gives the upper bound
\begin{equation*}
|\cC_\mu(\eta_\mu+w(t_0), \eta_\mu)| \le c'\|w(t_0)\|_{H^1}^2 , \ \ \mathrm{for }\ \  \|w(t_0)\|_{H^1}<1,
\end{equation*} 
where $c'$ is a constant independent of $h,\epsilon,$ see Remark \ref{rm:Nonlinearity}, Subsection \ref{sec:Assumptions}. Therefore,
\begin{equation*}
|\cC_\mu (\eta_\mu+ w(t),\eta_\mu)| \le C\|w(t_0)\|_{H^1}^2 + C(t-t_0) (h^2 \|w(t)\|_{H^1} + (|\alpha|+h) \|w(t)\|_{H^1}^2),
\end{equation*}
for some constant $C$ independent of $h$ and $\epsilon.$ Together with Lemma \ref{lm:LFLowerBound}, it follows that
\begin{equation*}
\frac{1}{2}\rho \|w(t)\|_{H^1}^2 \le C\|w(t_0)\|_{H^1}^2 + C(t-t_0) (h^2 \|w(t)\|_{H^1} + (|\alpha|+h) \|w(t)\|_{H^1}^2) + C\|w(t)\|_{H^1}^3,
\end{equation*}
where $\rho$ appears in (\ref{eq:Coercivity}). Equivalently, there exists a positive constant $C$ independent of $h$ and $\epsilon$ such that 
\begin{equation*}
C\|w(t)\|_{H^1} \le \|w(t_0)\|_{H^1}^2 + (t-t_0)(h^2 \|w(t)\|_{H^1} + (|\alpha|+h) \|w(t)\|_{H^1}^2 ) + \|w(t)\|_{H^1}^3.
\end{equation*}
For $t- t_0 \le \frac{C}{2(h+|\alpha|)}=:\tau,$
\begin{equation*}
C\|w(t)\|_{H^1}^2 \le \|w(t_0)\|^2_{H^1} + \frac{C}{2}h \|w(t)\|_{H^1} + \frac{C}{2}\|w(t)\|_{H^1}^2 + \|w(t)\|_{H^1}^3.
\end{equation*}
Using the fact that 
\begin{equation*}
h \|w(t)\|_{H^1} \le \frac{1}{2}h^2 + \frac{1}{2}\|w(t)\|_{H^1}^2,
\end{equation*}
we have
\begin{equation*}
\|w(t_0)\|_{H^1}^2 + \frac{C}{4} h^2 - \frac{C}{4} \|w(t)\|_{H^1}^2 + \|w(t)\|_{H^1}^3 \ge 0.
\end{equation*}
Let $y_0 := \|w(t_0)\|_{H^1},$ $y:= \sup_{t\in [t_0, t_0 +\tau ]}\|w(t)\|_{H^1},$ and 
$f(y) = y^3 - \frac{C}{4}y^2 + y_0^2 + \frac{C}{4}h^2.$ For $h\ll 1,$ the function intersects the x-axis in a point $y_*$ such that $y_*^2< c_1(h^2+ y_0^2),$ where $c_1$ is a positive constant independent of $h$ and $\epsilon;$ see Figure 1. It follows that for $y_0<y_*,$ $y<y_*,$ for $t\in [t_0, t_0+\tau].$ Substituting back in (\ref{eq:a}-\ref{eq:m}) and using (\ref{eq:Coefficients}) and (\ref{eq:Alpha}) gives 
\begin{equation*}
|\alpha| \le c_2 (h^2 + y_0^2),
\end{equation*}   
for some positive constant $c_2$ which is independent of $h$ and $\epsilon.$ It follows that for $h$ small enough, there exists positive constants $C_1$ and $C_2$ which are independent of $h$ and $\epsilon,$ such that
\begin{align*}
&\sup_{t\in [t_0,t_0 + \frac{C_2}{h}]} \|w(t)\|_{H^1}^2 \le C_1 (h^2 + \|w(t_0)\|^2_{H^1}) \\
&\sup_{\stackrel{t\in [t_0,t_0 + \frac{C_2}{h}]}{i\in \{1,\cdots , 2N+2\}}} |\alpha_i(t)| \le C_1 (h^2 + \|w(t_0)\|^2_{H^1}).
\end{align*} 
This concludes the proof.

\end{proof}

\begin{center}
\begin{picture}(0,0)%
\includegraphics{fig1.pstex}%
\end{picture}%
\setlength{\unitlength}{2693sp}%
\begingroup\makeatletter\ifx\SetFigFontNFSS\undefined%
\gdef\SetFigFontNFSS#1#2#3#4#5{%
  \reset@font\fontsize{#1}{#2pt}%
  \fontfamily{#3}\fontseries{#4}\fontshape{#5}%
  \selectfont}%
\fi\endgroup%
\begin{picture}(5559,3820)(2644,-5849)
\put(4096,-4921){\makebox(0,0)[lb]{\smash{{\SetFigFontNFSS{8}{9.6}{\rmdefault}{\mddefault}{\updefault}{\color[rgb]{0,0,0}$y_*$}%
}}}}
\put(7921,-4966){\makebox(0,0)[lb]{\smash{{\SetFigFontNFSS{8}{9.6}{\rmdefault}{\mddefault}{\updefault}{\color[rgb]{0,0,0}$y$}%
}}}}
\put(3061,-2266){\makebox(0,0)[lb]{\smash{{\SetFigFontNFSS{8}{9.6}{\rmdefault}{\mddefault}{\updefault}{\color[rgb]{0,0,0}$f(y)$}%
}}}}
\put(4951,-4921){\makebox(0,0)[lb]{\smash{{\SetFigFontNFSS{8}{9.6}{\rmdefault}{\mddefault}{\updefault}{\color[rgb]{0,0,0}$\frac{C}{6}$}%
}}}}
\end{picture}%

\end{center}

\subsection{Proof of Theorem \ref{th:Main}}

In this subsection, we prove Theorem \ref{th:Main} by iterating the application of Proposition \ref{pr:UpperBdFluctuation} and using the result of Proposition \ref{pr:RepEqMotion}.\footnote{This iteration scheme is similar in spirit to the one used in \cite{GusSi1}; see also \cite{HZ1}.}

\begin{proof}[Proof of Theorem \ref{th:Main}]
Fix $\epsilon\in (0,1).$ For $h$ small enough, let $T^*$ be the maximal time for which the skew-orthogonal decomposition is possible; see Subsection \ref{sec:SOD}. Consider the interval 
\begin{equation*}
[0,T] = [t_0,t_1]\cup [t_1,t_2] \cup \cdots \cup [t_{n-1},t_n] \subset [0,T^*],
\end{equation*}
such that 
$$0=t_0< t_1 < \cdots < t_n=T, \ \ (t_{i+1}-t_i) \le \frac{C_2}{h}, \ \ i=0,\cdots, n-1,$$ 
where $C_2$ appears is Proposition \ref{pr:UpperBdFluctuation}. We will choose $n\in {\mathbb N}$ depending on $h$ and $\epsilon$ later. Let 
\begin{align*}
|\alpha|_i &:= \sup_{\stackrel{t\in[t_i,t_{i+1}]}{j\in \{1,\cdots ,2N+2\}}}|\alpha_j| ,\\
y_j &:= \sup_{t\in [t_j,t_{j+1}]} \|w(t)\|_{H^1}, \ \ j=0,\cdots , n-1.
\end{align*}
Note that $y_0 \le h$ and $|\alpha|_0 \le C h,$ for some constant $C$ independent of $h$ and $\epsilon.$ Iterating the application of Proposition \ref{pr:UpperBdFluctuation} we have 
\begin{align*}
y_n^2 &\le (\sum_{j=1}^n C_1^j) h^2 \le C_1^{n+1} h^2  \\
|\alpha|_n &\le C C_1^{n+1}h^2.
\end{align*}
We choose $n$ such that $C_1^{n+1}h^2 \le h^{2-\epsilon}.$ This implies 
\begin{equation*}
n+1 \le -\epsilon \frac{\log h}{\log C_1}. 
\end{equation*}
Therefore, for $t\in [0, \epsilon \frac{C_2}{\log C_1} \frac{|\log h|}{h}],$
\begin{align*}
&\|w(t)\|_{H^1}^2 \le h^{2-\epsilon} \\
&|\alpha| \le C h^{2-\epsilon}.
\end{align*}
The claim of the theorem follows by the application of Proposition \ref{pr:RepEqMotion}.
\end{proof}


\section{Two physical applications}\label{sec:Applications}

In this section, we sketch two physical applications of our analysis. Throughout our discussion, we suppose Assumptions (A1)-(A7) are satisfied, so that the results of Theorem \ref{th:Main}, Section \ref{sec:Main}, hold.

\subsection{Adiabatic transportation of solitons}

We discuss adiabatic transportation of solitons in time-dependent confining potentials. Suppose the external potential is locally harmonic and decaying when $\|x\|\rightarrow \infty$ such that 
\begin{equation*}
\nabla V_h(t,x) = h^2\omega_0^2(x-st) , \ \ \|x-st\| \le \theta
\end{equation*}
where $s\in \bbR^N$ and $\theta\gg 1.$ Suppose the soliton center of mass is initially at $x=0.$
Choosing $\epsilon = \frac{1}{2}$ in Theorem \ref{th:Main}, we have, for $\theta \gg 1$ large enough, and some time $t< C|\log h|/h,$
\begin{align*}
\partial_t a &= v + O(h^{\frac{3}{2}}) ,\\
\partial_t v &= -\omega_0^2 h^2 (a-s t) + O(h^{\frac{3}{2}}).
\end{align*} 
Making the change of variables 
\begin{align*}
\tilde{a} &= a -st\\
\tilde{v} &= v-s, 
\end{align*}
the above equations become 
\begin{align*}
\partial_t \tilde{a} &= \tilde{v} + O(h^{\frac{3}{2}}) ,\\
\partial_t \tilde{v} &= -\omega_0^2 h^2 \tilde{a} + O(h^{\frac{3}{2}}),
\end{align*} 
whose solution is of the form 
\begin{equation*}
\tilde{a}(t) = a(0) \cos(h\omega_0 t) + \frac{v(0)-s}{h\omega_0} \sin(h\omega_0 t) + O(h^{\frac{3}{2}}).
\end{equation*}
For $a(0)=v(0)=0,$
\begin{equation*}
a(t) = st - \frac{s}{h\omega_0} \sin (h\omega_0 t)+ O(h^{\frac{3}{2}}), 
\end{equation*}
for $t\le C |\log h|/h.$ If $\frac{\|s\|}{h\omega_0}= O(h^\alpha), \alpha \ge \frac{3}{2},$ $a(t) = st + O(h^{\frac{3}{2}})$ $,\ie $ the soliton is transported adiabatically, up to error terms due to radiation damping and oscillations.

\subsection{Mathieu instability due to time-periodic perturbation of trapped solitons}

As a second application of our analysis, we discuss the onset of Mathieu instability of initially trapped solitons due to time periodic perturbations. We note that our analysis can be easily generalized to more general time-periodic potentials, \cite{MW1,Ar1}, and quasi-periodic potentials, \cite{Pa1}. 

Suppose the external potential is locally harmonic and decaying at spatial infinity such that
\begin{equation*}
\nabla V_h(t,x) = h^2\omega_0^2 (1 + \delta \cos(\omega t)) x , \ \ \|x\|\le \theta
\end{equation*}
where $\theta \gg 1.$ We claim that for special values of $\omega_0/\omega$ the system exhibits Mathieu instability, in the sense that the minimum of the potential becomes unstable under small perturbations. 
Choosing $\epsilon = \frac{1}{2}$ in Theorem \ref{th:Main}, we have, for $\theta$ large enough, and time $t< C|\log h|/h,$
\begin{align*}
\partial_t a &= v + O(h^{\frac{3}{2}}) ,\\
\partial_t v &= -\omega_0^2 (1 + \delta \cos(\omega t)) h^2 a + O(h^{\frac{3}{2}}).
\end{align*} 
Note that this is nothing but Mathieu's equation, plus error terms due to radiation damping. For more general periodic potentials, one obtains Hill's equation, \cite{MW1}. The system exhibits parametric resonance if $h\omega_0=n\frac{\omega}{2}, n=1,2,\cdots , \delta > 0,$ see for example \cite{Ar1}, Chapter 5. Although the nonlinear term acts as {\it friction}, the minimum of the potential is unstable for $\delta$ large enough, with $h^{\frac{1}{2}}\ll \delta \ll 1.$ We note that the analysis is restricted to the instability of the minimum of the potential. When $\|a\|$ grows to $O(h^{-\frac{3}{4}})$ due to the instability, neglecting the nonlinear terms in the equations of motion is no more justified, and one needs to analyse the full nonlinear problem, see for example \cite{BaKr1}. Instead, depending on the form of the potential, one can perceive different scenareos: The center of mass of the soliton oscillates chaotically, or there is a bifurcation and new points of stability. Further discussion of this application will appear elsewhere.

\section{Appendix}\label{sec:Appendix}

\noindent{\bf Rate of change of field energy and momenta, Section \ref{sec:ProofMain}}\label{app:Ehrenfest}

\noindent{\it Proof of (\ref{eq:RatePotential})}. Differentiating $\langle \psi, V \psi\rangle$ with respect to $t$ and using (\ref{eq:NLSE}), Assumption (A1), and the fact that $\psi\in H^1(\bbR^N),$ we have 
\begin{align*}
\partial_t \langle \psi, V \psi\rangle &= \langle \psi, \partial_t V \psi \rangle + \langle \partial_t \psi, V \psi \rangle  + \langle \psi, V \partial_t\psi\rangle \\ 
&= \langle \psi, \partial_t V \psi \rangle + \langle -\Delta \psi + V \psi -f(\psi), iV\psi\rangle \\ &+ \langle i\psi , V (-\Delta \psi + V\psi -f(\psi))\rangle \\
&= \langle \psi, \partial_t V \psi \rangle + 2\langle i\nabla V\psi, \nabla\psi\rangle ,
\end{align*}
where we have used integration by parts in the last step. 

\noindent{\it Proof of (\ref{eq:RateEnergy})}. The proof of (\ref{eq:RateEnergy}) follows directly from (\ref{eq:RateH1}) in the proof of Proposition \ref{pr:EnergyUpperBd}, Section \ref{sec:Well-Posedness}, with the identification
\begin{equation*}
g(t,u) = V(t)u - f(u).
\end{equation*}

\noindent{\it Proof of (\ref{eq:Ehrenfest})}. We use a regularization scheme similar to the one used in Proposition \ref{pr:EnergyUpperBd}, Section \ref{sec:Well-Posedness}. Let $I_\epsilon := (1-\epsilon\Delta)^{-1}.$
\begin{align*}
&\partial_t \langle i\psi ,\nabla \psi \rangle = \partial_t \lim_{\epsilon\searrow 0} \langle I_\epsilon i \psi, I_\epsilon \nabla \psi \rangle \\
&= \lim_{\epsilon\searrow 0}\{ \langle I_\epsilon i\partial_t \psi , I_\epsilon \nabla\psi \rangle + \langle I_\epsilon i\psi, I_\epsilon \nabla \partial_t \psi\rangle \} \\
&= \lim_{\epsilon \searrow 0} \{ \langle I_\epsilon (-\Delta \psi + V\psi -f(\psi)), I_\epsilon \nabla \psi\rangle  - \langle I_\epsilon \psi, I_\epsilon \nabla (-\Delta \psi + V\psi -f(\psi)) \rangle \}\\
&= \langle \psi, V \nabla \psi\rangle - \langle \psi, \nabla (V\psi)\rangle \\
&= - \langle \psi, \nabla V\psi\rangle .
\end{align*}



\begin{thebibliography}{9}

\bibitem{FTY1}
J.~Fr\"ohlich, T.-P. Tsai, and H.-T. Yau.
\newblock On a classical limit of quantum theory and the non-linear Hartree
  equation.
\newblock {\em Geom. Funct. Anal.}, Special Volume: 57-78, 2000.

\bibitem{FTY2}
J.~Fr\"ohlich, T.-P. Tsai, and H.-T. Yau.
\newblock On the point-particle (Newtonian) limit of the non-linear Hartree equation.
\newblock {\em Commun. Math. Phys.}, 225(2): 223-274, 2002.

\bibitem{BJ1}
J.~C. Bronski and R.~L. Jerrard.
\newblock Soliton dynamics in a potential.
\newblock {\em Math. Res. Lett.,} 7(2-3): 329-342, 2000.

\bibitem{FJGS1}
J.~Fr\"ohlich, S.~Gustafson, B.~L.~G. Jonsson, and I.~M. Sigal.
\newblock Solitary wave dynamics in an external potential.
\newblock {\em Commun. Math. Phys.}, 250(3): 613-642, 2004.

\bibitem{FJGS2}
J.~Fr\"ohlich, S.~Gustafson, B.~L.~G. Jonsson, and I.~M. Sigal.
\newblock Long time motion of NLS solitary waves in a confining potential. \newblock {\em Annals Henri Poincare,} 7: 621-660, 2006.


\bibitem{GS1}
Z.~Gang and I.~M. Sigal.
\newblock On Soliton Dynamics in Nonlinear Schr\"odinger Equations.
\newblock {\em GAFA,} 2007, to appear.

\bibitem{HZ1}
J. Holmer and M. Zworski.
\newblock Slow soliton interaction with delta impurities.
\newblock {\it Preprint} 2007.



\bibitem{Str1}
R.S. Strichartz.
\newblock Restrictions of Fourier transforms to quadratic surfaces and decay of solutions of wave equations.
\newblock {\em Duke Math. J.}, 44: 705-714, 1977.

\bibitem{St1}
W.~A. Strauss.
\newblock Existence of solitary waves in higher dimensions. 
\newblock {\em Commun. Math. Phys.,} 55(2): 149-162, 1977.

\bibitem{GV1}
J. Ginibre and G. Velo.
\newblock On a class of nonlinear Schr\"odinger equations. I,II.
\newblock {\em J. Func. Anal.,} 32: 1-71, 1979.

\bibitem{GV2}
J. Ginibre and G. Velo.
\newblock On a class of nonlinear Schr\"odinger equations with nonlocal interactions. 
\newblock {\em Math. Z.,} 170(2): 109-136, 1980.

\bibitem{GV3}
J. Ginibre and G. Velo.
\newblock The global Cauchy problem for the nonlinear Schr\"odinger equation revisited.
\newblock {\em Ann. Inst. H. Poincar\'e Anal. Non Lin\'eaire.}, 2: 309-327, 1985.

\bibitem{KeTa1}
M. Keel and T. Tao.
\newblock Endpoint Strichartz inequalities.
\newblock {\em Amer. J. Math.}, 120: 955-980, 1998.

\bibitem{BL1}
H.~Berestycki and P.-L. Lions.
\newblock Nonlinear scalar field equations. I. Existence of a ground state.
\newblock {\em Arch. Rational Mech. Anal.}, 82(4): 347-375, 1983.

\bibitem{BL2}
H.~Berestycki and P.-L. Lions.
\newblock Nonlinear scalar field equations. II. Existence of infinitely many solutions.
\newblock {\em Arch. Rational Mech. Anal.}, 82(4): 347-375, 1983.

\bibitem{BLP1}
H.~Berestycki, P.-L. Lions, and L.~A. Peletier.
\newblock An {ODE} approach to the existence of positive  solutions for semilinear problems in {$\mathbb{R}^{N}$}.
\newblock {\em Indiana Univ. Math. J.}, 30(1): 141-157, 1981.


\bibitem{GSS1}
M.~Grillakis, J.~Shatah, and W.~Strauss.
\newblock Stability theory of solitary waves in the presence of symmetry. {I}.
\newblock {\em J. Funct. Anal.}, 74(1): 160-197, 1987.

\bibitem{GSS2}
M.~Grillakis, J.~Shatah, and W.~Strauss.
\newblock Stability theory of solitary waves in the presence of symmetry. {II}.
\newblock {\em J. Funct. Anal.}, 94(2): 308-348, 1990.


\bibitem{We1}
M.~I. Weinstein.
\newblock Modulational stability of ground states of nonlinear Schr\"{o}dinger equations.
\newblock {\em SIAM J. Math. Anal.}, 16(3): 472-491, 1985.

\bibitem{We2} 
M.~I. Weinstein.
\newblock Lyapunov stability of ground states of the nonlinear dispersive evolution equations.
\newblock {\em Commun. Pure Appl. Math.,} XXXIX:51-68, 1986.

\bibitem{BP1}
V.~S. Buslaev and G.~S. Perel'man.
\newblock On the stability of solitary waves for nonlinear Schr\"{o}dinger
  equations.
\newblock {\em Amer. Math. Soc. Transl. Ser.}, 2(164): 74-98, 1995.

\bibitem{BS1}
V.~S. Buslaev and C.~Sulem.
\newblock On asymptotic stability of solitary waves for nonlinear
  Schr{\"o}dinger equations.
\newblock {\em Ann. IHP. Analyse Nonlin\'eaire}, 20: 419-475, 2003.



\bibitem{Cu1}
S.~Cuccagna.
\newblock Stabilization of solutions to nonlinear Schr\"{o}dinger equations.
\newblock {\em Commun. Pure Appl. Math.}, 54(9): 1110-1145, 2001.

\bibitem{Cu2}
S.~Cuccagna.
\newblock Asymptotic stability of the ground states of the nonlinear
  Schr{\"o}dinger equation.
\newblock {\em Rend. Istit. Mat. Univ. Trieste}, 32: 105-118, 2002.

\bibitem{Ka1}
T.~Kato
\newblock On nonlinear Schr\"odinger equations.
\newblock {\em Ann. Inst. H. Poincar\'e Phys. Th\'eor.}, 46: 113-129, 1987. 

\bibitem{Ka2}
T.~Kato
\newblock On nonlinear Schr\"odinger equations. II. $H^s$-solutions and unconditional well-posedness.
\newblock {\em J. Anal. Math.}, 67: 281-306, 1995.

\bibitem{Ca1}
T. Cazenave.
\newblock {\em An Introduction to Nonlinear Schr\"odinger Equations.}
\newblock Textos de M\'etodos Matem\'aticos 26. Instituto de Matem\'atica, Rio de Janeiro, 1996.

\bibitem{SS1} 
C. Sulem and P.-L. Sulem.
\newblock {\em The Nonlinear Schr\"odinger Equation.}
\newblock Number 130 in Applied Mathematical Sciences. Springer, New York, 1999.


\bibitem{BerLoef1}
J.~Bergh and J.~L\"ofstr\"om.
\newblock {\em Interpolation Spaces.}
\newblock Springer, New York, 1976.

\bibitem{Oh1}
Y.-G. Oh.
\newblock Cauchy problem and Ehrenfest's law of nonlinear Schr\"odinger equations with potentials. 
\newblock {\em J. Differential Equations}, 81: 255-274, 1989.



\bibitem{CaWe1}
T.~Cazenave and F.~B. Weissler.
\newblock Rapidly decaying solutions of the nonlinear Schr\"odinger equation.
\newblock {\em Commun. Math. Phys.}, 147: 75-100, 1992.


\bibitem{CKP1}
S. Cuccagna, E. Kirr and D. Pelinovsky.
\newblock Parametric resonance of ground states in the nonlinear Schr\"odinger equation.
\newblock {\em J. Differential Equations,} 220(1): 85-120, 2006.



\bibitem{SR1}
I. Rodnianski and W. Schlag.
\newblock Time decay for solutions of Schrödinger equations with rough and time-dependent potentials.  
\newblock {\em Invent. Math.,} 155(3): 451-513, 2004.


\bibitem{RSS1}
I. Rodnianski, W. Schlag and A. Soffer.
\newblock Dispersive analysis of charge transfer models.
\newblock {\em Commun. Pure Appl. Math.}, 58: 149-216, 2005.


\bibitem{RSS2}
I. Rodnianski, W. Schlag and A. Soffer.
\newblock Asymptotic stability of N-soliton states of NLS. Preprint.

\bibitem{Ya1}
K. Yajima.
\newblock Existence of solutions for Schr\"odinger evolution equations.
\newblock {\em Commun. Math. Phys.}, 110: 415-426, 1987.


\bibitem{Bo1}
J. Bourgain.
\newblock Growth of Sobolev norms in linear Schr\"odinger equations with quasi-periodic potentials. 
\newblock {\em Commun. Math. Phys.}, 204: 207-247, 1999.

\bibitem{GusSi1}
S. Gustafson and I.M. Sigal.
\newblock Dynamics of magnetic vortices.
\newblock {\em Adv. Math.}, 199: 448-498, 2006.

\bibitem{Ar1}
V.~Arnold 
\newblock {\em Mathematical Methods of Classical Mechanics}.
\newblock Springer-Verlag, Graduate texts in Mathematics, New York, 1984. 

\bibitem{MW1}
M. Magnus and S. Winkler.
\newblock {\em Hill's Equation}.
\newblock Dover Publications, New York, 2004.

\bibitem{Pa1}
I.O.~Parasyuk.
\newblock Parametric resonance in linear Hamiltonian systems with quasi-periodic coefficients.
\newblock {\em Ukranian Math. J.}, 31: 73-75, 1979.

\bibitem{BaKr1}
V.I. Babitsky and V.L. Krupenin.
\newblock {\em Vibrations of Strongly Nonlinear Systems}.
\newblock Springer-Verlag, New York, 2001.

\end{thebibliography}
\end{document}